%% file: root.tex
\documentclass[journal,twoside,web]{ieeecolor}
\usepackage{generic}
\usepackage{cite}
\usepackage{amsmath,amssymb,amsfonts,mathtools,nicefrac}
\usepackage{algorithmic}
\usepackage{graphicx}
\usepackage{textcomp}
\def\BibTeX{{\rm B\kern-.05em{\sc i\kern-.025em b}\kern-.08em
    T\kern-.1667em\lower.7ex\hbox{E}\kern-.125emX}}
\usepackage[tracking=false,kerning=true,spacing=true]{microtype}
\renewcommand{\baselinestretch}{0.98}
\usepackage[utf8]{inputenc}
\usepackage[english]{babel}
\usepackage[usenames,dvipsnames]{xcolor}

\usepackage[acronym,nomain,nonumberlist]{glossaries}
\let\labelindent\relax
\usepackage{enumerate,enumitem}
\usepackage{booktabs,multirow,makecell}
\usepackage{threeparttable}
\usepackage[aboveskip=0pt,belowskip=0pt,format=plain,textformat=simple,textfont=footnotesize,labelfont={bf,footnotesize},labelsep=colon]{caption}
\usepackage{subcaption}
\usepackage{adjustbox}
\usepackage[gen]{eurosym}
\usepackage[hidelinks]{hyperref}
\usepackage[capitalise,noabbrev]{cleveref}
\usepackage{pgfplots}
\usepackage{placeins}
\usepackage{pifont}
\usepackage{ifthen,xifthen}
\usepackage[per-mode=symbol]{siunitx}
\usepackage[skins]{tcolorbox}
\usepackage{packages/pgf-pie/pgf-pie}

\AtEndEnvironment{figure}{\vspace{-0.5cm}}
\AtEndEnvironment{table}{\vspace{-0.5cm}}

\newtheorem{definition}{Definition}
\newtheorem{theorem}{Theorem}
\newtheorem{lemma}{Lemma}
\newtheorem{proposition}{Proposition}

\newtheorem{problem}{Problem}
\newtheorem{remark*}{Remark}

\usepackage{nomencl,multicol}
\renewcommand{\nompreamble}{\begin{multicols}{2}}
\renewcommand{\nompostamble}{\end{multicols}}
\makenomenclature
\renewcommand{\nomgroup}[1]{
  \ifthenelse{\equal{#1}{C}}{\item[\textbf{Demand-related symbols}]}{
  \ifthenelse{\equal{#1}{B}}{\item[\textbf{Operators-related symbols}]}{
  \ifthenelse{\equal{#1}{A}}{\item[\textbf{Graph-related symbols}]}{
  \ifthenelse{\equal{#1}{D}}{\item[\textbf{Other symbols}]}{}}}}
}

\hypersetup{pdfauthor={Lanzetti, Schiffer, Ostrovsky, Pavone},
            pdftitle={On the Interplay between Self-Driving Cars and Public Transportation}}

\input{./chapters/basics.tex}

\begin{document}

\title{On the Interplay between Self-Driving Cars and Public Transportation}
\author{Nicolas Lanzetti, Maximilian Schiffer, Michael Ostrovsky, and Marco Pavone
\thanks{Ostrovsky is grateful to the National Science Foundation for financial support (grant SES-1824317).
Pavone was supported in part by the National Science Foundation, CAREER Award CMMI 1454737.}
\thanks{Nicolas Lanzetti is with the Automatic Control Laboratory, ETH Zürich, Zürich (\url{lnicolas@ethz.ch}). }
\thanks{Maximilian Schiffer is with the TUM School of Management and the Munich Data Science Institute, TU Munich, Munich (\url{schiffer@tum.de})}
\thanks{Michael Ostrovsky is with the Graduate School of Business, Stanford University, Stanford and the National Bureau of Economic Research, Cambridge (\url{ostrovsky@stanford.edu}).}
\thanks{Marco Pavone is with the Department of Aeronautics and Astronautics, Stanford University, Stanford (\url{pavone@stanford.edu}).}
}

\input{chapters/0_titlepage.tex}

\input{chapters/acronyms.tex}

\input{chapters/1_introduction.tex}
\input{chapters/2_problemsetting.tex}
\input{chapters/3_methodology.tex}
\input{chapters/4_specialcase.tex}
\input{chapters/5_casestudy.tex}
\input{chapters/6_results.tex}
\input{chapters/7_conclusion.tex}

\appendices
\input{chapters/appendixnotation.tex}
\input{chapters/appendixproofs.tex}


\bibliographystyle{IEEEtran}
\bibliography{main}

\vspace{-0.7cm}

\begin{IEEEbiography}[{\includegraphics[width=1in,height=1.25in,clip,keepaspectratio]{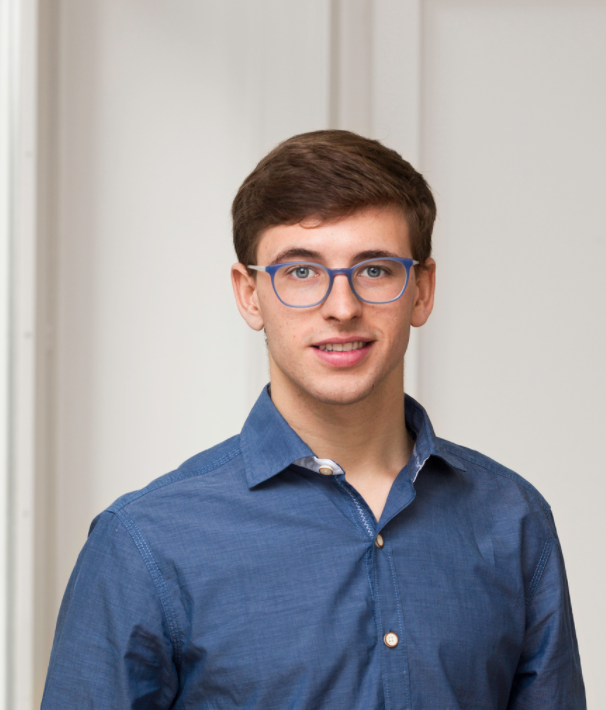}}]{Nicolas Lanzetti}
is a Ph.D. candidate at the Automatic Control Laboratory at ETH Zurich, supervised by Prof. Florian D\"orfler.
He received the B.Sc. and the M.Sc. degrees in mechanical engineering, with focus on Robotics, Systems, and Control from ETH Zürich in 2016 and 2019, respectively.
He was a visiting researcher at Massachusetts Institute of Technology and Stanford University.
His research interests include optimal transport and Wasserstein gradient flows, with applications to optimization and game theory.
\end{IEEEbiography}

\vspace{-0.7cm}

\begin{IEEEbiography}[{\includegraphics[width=1in,height=1.25in,clip,keepaspectratio]{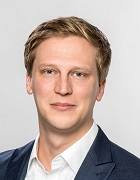}}]{Maximilian Schiffer}
 is an Associate Professor of Business Analytics \& Intelligent Systems and a core member of the Munich Data Science Institute at the Technical University of Munich. He received a Ph.D. from RWTH Aachen University in 2017. His areas of expertise are in operations research and (interpretable) machine learning. His research interests include a wide range of transportation and logistics topics, such as electric
vehicles and autonomous systems, but also topics from production planning, supply chain management, and data science.
\end{IEEEbiography}

\vspace{-0.7cm}

\begin{IEEEbiography}[{\includegraphics[width=1in,height=1.25in,clip,keepaspectratio]{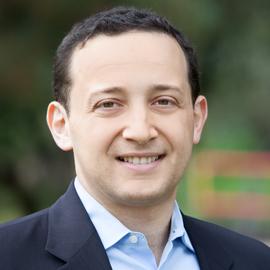}}]{Michael Ostrovsky}
is the Fred H. Merrill Professor of Economics at the Graduate School of Business at Stanford University and a Co-Director of the Market Design Working Group at the National Bureau of Economic Research. He received his Ph.D. in Business Economics from Harvard University in 2005. His research interests are in the areas of game theory, market design, industrial organization, and finance. In his recent research, he has analyzed the economics of carpooling and self-driving cars, the properties of internet advertising auctions, information aggregation in financial markets, stability in trading networks, and voting in shareholder meetings.
\end{IEEEbiography}

\vspace{-0.7cm}

\begin{IEEEbiography}[{\includegraphics[width=1in,height=1.25in,clip,keepaspectratio]{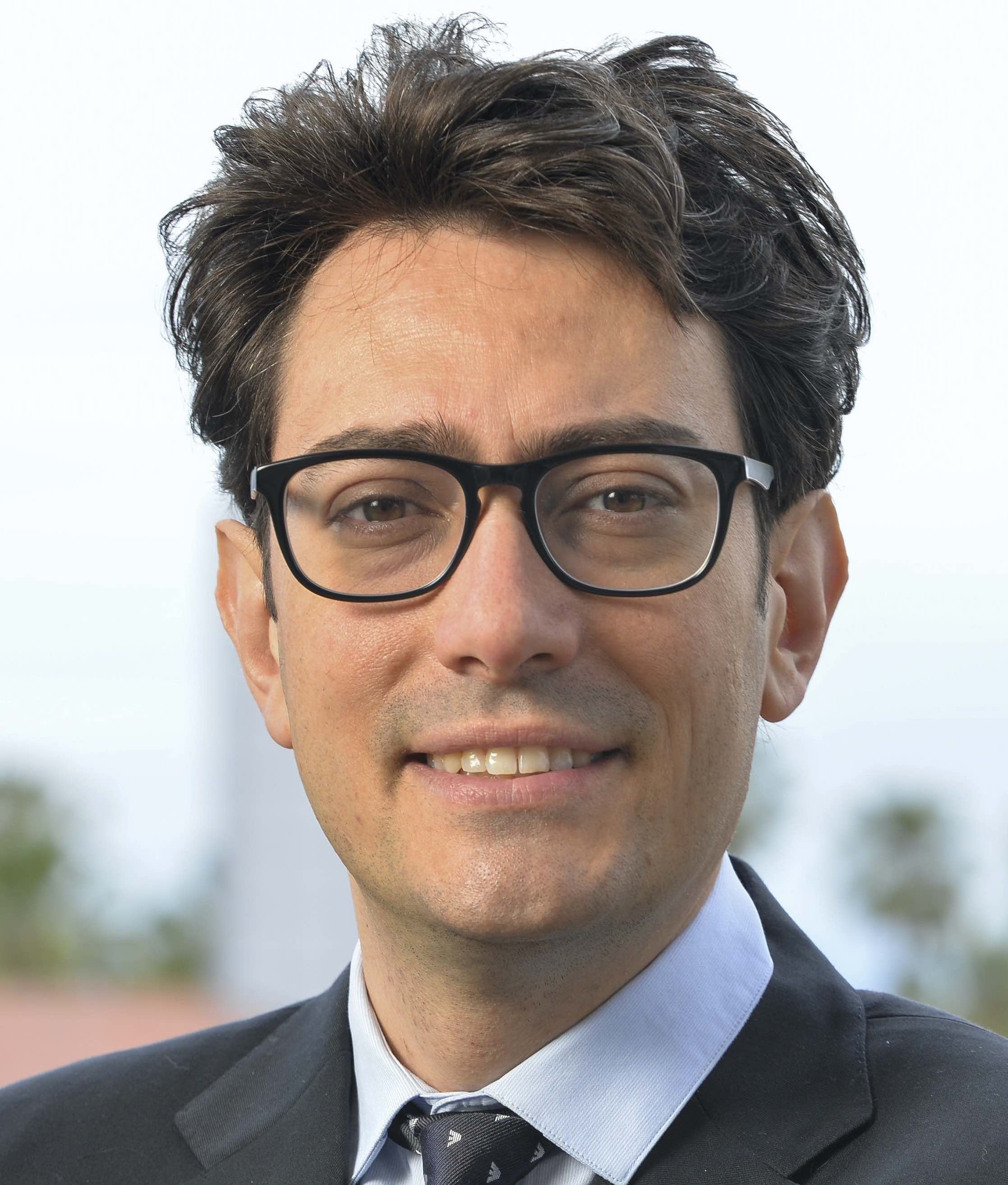}}]{Marco Pavone} is an Associate Professor of Aeronautics and Astronautics at Stanford University and the Director of Autonomous Vehicle Research at NVIDIA. He received a Ph.D. degree in Aeronautics and Astronautics from the Massachusetts Institute of Technology in 2010. His main research interests are in the development of methodologies for the analysis, design, and control of autonomous systems, with an emphasis on self-driving cars, autonomous aerospace vehicles, and future mobility systems.
\end{IEEEbiography}

\end{document}

%% file: chapters/basics.tex
\usetikzlibrary{calc,positioning,external,shapes}
\usepgfplotslibrary{external,groupplots,statistics}
\pgfplotsset{every axis/.append style={tick label style={font=\scriptsize},
                                       legend style={font=\scriptsize}, 
                                       /pgf/number format/use comma,
	                                   },
	         every axis/.append style={font={\footnotesize}},
	         compat=1.15,
             }        
\graphicspath{
	{./graphics/}}

\tikzset{declare function={a=0.140002;
	erf(\x)=sign(\x)*sqrt(1-exp(-\x*\x*(((4/pi)+a*\x*\x)/(1+a*\x*\x))));}}

\providecommand\arc{}
\providecommand\demand{}
\providecommand\argmin{}
\providecommand\argmax{}

\newcommand\normalFont[2][]{
    \ifthenelse{\equal{#1}{tilde}}
    {\tilde{#2}}
    {\ifthenelse{\equal{#1}{bar}}
                {\bar{#2}}
                {\ifthenelse{\equal{#1}{star}}
                            {{#2}^\star}
                            {#2}
                }
    }
}

\newcommand\mathcalFont[2][] {\normalFont[#1]{\mathcal{#2}}}

\definecolor{lightblue}    {rgb}{0.60784,0.76078,0.90196}
\definecolor{darkblue}     {rgb}{0.26667,0.44706,0.76863}
\definecolor{lightgreen}   {rgb}{0.66275,0.81569,0.55686}
\definecolor{darkgreen}    {rgb}{0.43922,0.67843,0.27843}
\definecolor{orange}       {rgb}{0.92941,0.49020,0.19216}
\definecolor{yellow}       {rgb}{1.00000,0.75294,0.00000}
\definecolor{grey}         {rgb}{0.64706,0.64706,0.64706}
\definecolor{purple}	   {rgb}{0.51373,0.23529,0.44706}
\definecolor{silver}	   {rgb}{0.75294,0.75294,0.79254}

\definecolor{amod0}         {rgb}{0      , 0.4    , 0.7   }
\colorlet{amod}            {amod0!70!white}
\colorlet{amod1}           {amod0!50!white}
\definecolor{publicTransit}{rgb}{0.85   , 0.3    , 0.1   }
\colorlet{publicTransit}   {publicTransit!70!white}
\definecolor{walking}      {rgb}{0.9    , 0.7    , 0.1   }
\colorlet{walking}         {walking!70!white}
\definecolor{revenue}      {rgb}{1      , 0.8    , 0.2   }
\definecolor{cost}         {rgb}{1      , 0.2    , 0.2   }
\definecolor{profit}       {rgb}{0.25   , 0.55   , 0.25  }
\definecolor{fares}        {rgb}{0.0    , 0.6    , 0.0   }
\definecolor{taxes}        {rgb}{0.0    , 0.0    , 0.6   }

\DeclareSIUnit\eur{\euro}
\DeclareSIUnit\usd{USD}
\DeclareSIUnit{\mph}{mph}
\DeclareSIUnit\month{month}
\DeclareCaptionLabelFormat{andtable}{#1~#2  \&  \tablename~\thetable}

\renewcommand\argmin[1]{\underset{\substack{#1}}{\arg\min\,}}
\renewcommand\argmax[1]{\underset{\substack{#1}}{\arg\max\,}}
\newcommand\cardinality[1]{|#1|}
\newcommand\expectedValue[2][] {\mathbb{E}_{#1}\left[#2\right]}

\newcommand{\summe}[1]{\sum_{\scriptstyle\mathclap{\substack{#1}}}}

\makeatletter
    \if@cref@capitalise
        \crefname{subsection}{Section}{Sections}
        \crefname{subsubsection}{Section}{Sections}
        \crefname{assumption}{Assumption}{Assumptions}
        \crefname{problem}{Problem}{Problems}
    \else
        \crefname{subsection}{section}{sections}
        \crefname{subsubsection}{section}{sections}
        \crefname{assumption}{assumption}{assumptions}
        \crefname{problem}{problem}{problems}
    \fi
\makeatother

%% file: chapters/0_titlepage.tex
\maketitle

\begin{abstract}
Cities worldwide struggle with overloaded transportation systems and their externalities. The emerging autonomous transportation technology has the potential to alleviate these issues, but the decisions of profit-maximizing operators running large autonomous fleets could negatively impact other stakeholders and the transportation system. An analysis of these tradeoffs requires modeling the modes of transportation in a unified framework. In this paper, we propose such a framework, which allows us to study the interplay among mobility service providers (MSPs), public transport authorities, and customers. Our framework combines a graph-theoretic network model for the transportation system with a game-theoretic market model in which MSPs are profit maximizers while customers select individually-optimal transportation options. We apply our framework to data for the city of Berlin and present sensitivity analyses to study parameters that MSPs or municipalities can strategically influence. We show that autonomous ride-hailing systems may cannibalize a public transportation system, serving between 7\,\% and 80\,\% of all customers, depending on market conditions and policy restrictions.
\end{abstract}


%% file: chapters/acronyms.tex
\newacronym{abk:amod}{AMoD}{autonomous mobility-on-demand}
\newacronym{abk:msp}{MSP}{mobility service provider}
\newacronym[longplural={public transport authorities}]{abk:pta}{PTA}{public transport authority}
\def\road {\mathrm{AMoD}}
\newcommand\bestResponse[1] {B_{#1}}
\newcommand\demandCost[2][] {\normalFont[#1]{J}_{#2}}\nomenclature[c17]{$\demandCost{}$}{Cost related to a demand }
\newcommand\demandNumber {\normalFont{M}}\nomenclature[c13]{$\demandNumber$}{Number of demands}
\newcommand\maxNumberVehicles[1]   {N_{#1}^\mathrm{veh}}\nomenclature[d]{$\maxNumberVehicles{}$}{Number of vehicles}
\newcommand\agentsNumber[2][]      {\normalFont[#1]{N}_{#2}}
\newcommand\operatorsNumber[1][]   {\agentsNumber[#1]{}}\nomenclature[b13]{$\operatorsNumber$}{Number of operators}
\def\publicTransit {\mathrm{PTW}}
\newcommand\utility[1][] {\normalFont[#1]{U}}\nomenclature[b31]{$U$}{Profit}
\def\valueTime    {V_\mathrm{T}}\nomenclature[d]{$\valueTime$}{Value of time}
\def\valueTimeMin {V_\mathrm{T}^\mathrm{min}}
\def\valueTimeMax {V_\mathrm{T}^\mathrm{max}}
\renewcommand\arc[2][] {\normalFont[#1]{a}_{#2}}\nomenclature[a22]{$\arc{}$}{Arc}
\newcommand\operationCostFunction[2][] {\normalFont[#1]{c}_{#2}}\nomenclature[b34]{$\operationCostFunction{}(a)$}{Cost of arc $a$}
\newcommand\operationCostDistance[2][] {\normalFont[#1]{c}_{#2}^\mathrm{d}}\nomenclature[b35]{$\operationCostDistance{}$}{Distance-based cost}
\newcommand\destinationVertex[2][] {\normalFont[#1]{d}_{#2}}\nomenclature[c15]{$\destinationVertex{}$}{Destination vertex}
\newcommand\flow[2][] {\normalFont[#1]{f}_{#2}}\nomenclature[a41]{$\flow{}$}{Flow}
\newcommand\originVertex[2][]    {\normalFont[#1]{o}_{#2}}\nomenclature[c14]{$\originVertex{}$}{Origin vertex}
\newcommand\operatorsObject[1][] {#1}
\renewcommand\path[2][]            {\normalFont[#1]{p}_{#2}}\nomenclature[a31]{$\path{}$}{Path}
\newcommand\pathRoad[2][]          {\path[#1]{#2}^{\road}}
\newcommand\pathPublicTransit[2][] {\path[#1]{#2}^{\publicTransit}}
\renewcommand\demand[2][] {\normalFont[#1]{q}_{#2}}\nomenclature[c11]{$\demand{}$}{Demand}
\newcommand\request[2][] {\normalFont[#1]{r}_{#2}}
\newcommand\arcSource[2][]  {\normalFont[#1]{s}_{#2}^\mathrm{o}}\nomenclature[a14]{$\arcSource{}(\arc{})$}{Source of arc $\arc{}$}
\newcommand\arcTarget[2][]  {\normalFont[#1]{s}_{#2}^\mathrm{d}}\nomenclature[a15]{$\arcTarget{}(\arc{})$}{Sink of arc $\arc{}$}
\newcommand\timeArc[2][]           {\normalFont[#1]{t}_{#2}}\nomenclature[d31]{$t(\arc{})$}{Travel time to traverse arc $\arc{}$}
\newcommand\timeFunction[2][]      {\normalFont[#1]{t}_{#2}}
\newcommand\timeRoad[2][]          {\timeFunction[#1]{#2}^{\road}}\nomenclature[d32]{$\timeRoad{}$}{Travel time of the \gls{abk:amod} ride}
\newcommand\timePublicTransit[2][] {\timeFunction[#1]{#2}^{\publicTransit}}\nomenclature[d33]{$\timePublicTransit{}$}{Travel time of the public transport ride}
\newcommand\vertex[2][] {\normalFont[#1]{v}_{#2}}\nomenclature[a21]{$\vertex{}$}{Vertex}
\newcommand\arcsSet[2][] {\mathcalFont[#1]{A}_{#2}}\nomenclature[a13]{$\arcsSet{}$}{Set of arcs}
\newcommand\customersActionSpace[2][] {\mathcalFont[#1]{B}_{#2}^\mathrm{c}}\nomenclature[c22]{$\customersActionSpace{}$}{Customers' action space}
\newcommand\operatorsActionSpace[2][] {\mathcalFont[#1]{B}_{#2}^\mathrm{o}}\nomenclature[b42]{$\operatorsActionSpace{}$}{Operator's action space}
\newcommand\customersEquilibria[2][] {\mathcalFont[#1]{E}_{#2}}\nomenclature[c33]{$\customersEquilibria{}$}{Set of optimal reactions}
\newcommand\flowsSetGraph[2][] {\mathcalFont[#1]{F}({#2})}
\newcommand\flowsSet[2][]      {\mathcalFont[#1]{F}_{#2}}\nomenclature[a42]{$\flowsSet{}$}{Set of flows}
\newcommand\graph[2][]           {\mathcalFont[#1]{G}_{#2}}\nomenclature[a11]{$\graph{}$}{Multigraph}
\newcommand\graphDefinition[2][] {\graph[#1]{#2}=(\verticesSet[#1]{#2},\arcsSet[#1]{#2},\arcSource[#1]{#2},\arcTarget[#1]{#2})}
\newcommand\flowsSetDemand[2][] {\mathcalFont[#1]{H}_{#2}}
\nomenclature[b51]{$\flowsSetDemand{}(\reactionCurve{})$}{Set of potentially active sets of flows for the reaction curve $\reactionCurve{}$}
\newcommand\agentsSet[2][]    {\mathcalFont[#1]{N}_{#2}}
\newcommand\pathsSet[2][]     {\mathcalFont[#1]{P}(#2)}\nomenclature[a32]{$\pathsSet{\arcsSet{}}$}{Set of all paths for the arc set $\arcsSet{}$}
\newcommand\shortestPath[2][] {\mathcalFont[#1]{P}_{#2}^{*}}
\nomenclature[a35]{$\shortestPath{}(\originVertex{},\destinationVertex{})$}{Shortest path between $\originVertex{}$ and $\destinationVertex{}$}
\newcommand\demandsSet[2][] {\mathcalFont[#1]{Q}_{#2}}\nomenclature[c12]{$\demandsSet{}$}{Set of demands}
\newcommand\pathsSetRequest[2][] {\mathcalFont[#1]{S}(#2)}\nomenclature[c2]{$\pathsSetRequest{\demand{}}$}{Set of paths satisfying demand $\demand{}$}
\newcommand\verticesSet[2][] {\mathcalFont[#1]{V}_{#2}}\nomenclature[a12]{$\verticesSet{}$}{Set of vertices}
\newcommand\priceStrategiesSet[2][]           {\normalFont[#1]{\Xi}_{#2}}\nomenclature[b22]{$\priceStrategiesSet{}$}{Set of pricing strategies}
\newcommand\priceStrategiesSetRestricted[2][] {\normalFont[#1]{\Xi}_{\mathrm{res},#2}}
\newcommand\demandRate[2][] {\normalFont[#1]{\alpha}_{#2}}\nomenclature[c16]{$\demandRate{}$}{Demand rate}
\newcommand\flowRate[2][] {\normalFont[#1]{\beta}_{#2}}
\newcommand\priceStrategy[2][]      {\normalFont[#1]{\xi}_{#2}}\nomenclature[b21]{$\priceStrategy{}$}{Pricing strategy}
\newcommand\pricePublicTransit[2][] {\priceStrategy[#1]{i}^{\publicTransit}}
\def\letterProjection 		        {\pi}
\newcommand\originPath[1][]         {\normalFont[#1]{\letterProjection_\mathrm{o}}}\nomenclature[a33]{$\originPath(\path{})$}{Origin of path $\path{}$}
\newcommand\destinationPath[1][]    {\normalFont[#1]{\letterProjection_\mathrm{d}}}\nomenclature[a34]{$\destinationPath(\path{})$}{Destination of path $\path{}$}
\newcommand\projectionPathFlow[2][] {\normalFont[#1]{\letterProjection_{\mathrm{p}}}(#2)}\nomenclature[a42]{$\projectionPathFlow{\flow{}}$}{Path of flow $\flow{}$}
\newcommand\projectionRateFlow[2][] {\normalFont[#1]{\letterProjection_{\mathrm{r}}}(#2)}\nomenclature[a43]{$\projectionRateFlow{\flow{}}$}{Rate of flow $\flow{}$}
\newcommand\reactionCurve[2][] {\normalFont[#1]{\phi}_{#2}}\nomenclature[c21]{$\reactionCurve{}$}{Reaction curve}
\newcommand{\myast}{\textsuperscript{$\dagger$}}
\newcommand\reals					{\mathbb{R}}
\newcommand\naturals				{\mathbb{N}}
\newcommand\nonnegativeReals		    {\reals_{\geq 0}}
\newcommand\positiveReals			{\reals_{>0}}
\newcommand\arcsNumber[2][]     {\cardinality{\arcsSet[#1]{#2}}}
\newcommand\verticesNumber[2][] {\cardinality{\verticesSet[#1]{#2}}}
\newcommand\demandIndexSet         {\{1,\ldots,\demandNumber\}}
\newcommand\agentsIndexSet[2][]    {\{1,\ldots,\agentsNumber[#1]{#2}\}}
\newcommand\operatorsIndexSet[1][] {\agentsIndexSet[#1]{}}
\newcommand\equilibriumGeneral[1][]  {(\{\priceStrategy[#1]{j}\}_{j=1}^{\operatorsNumber})}
\newcommand\equilibriumAmod[1][]     {(\priceStrategy[#1]{1},\priceStrategy[#1]{2})}

%% file: chapters/1_introduction.tex
\section{Introduction}\label{sec:introduction}
Worldwide, cities struggle with overstrained transportation systems whose externalities cause economic and environmental harm such as working hours lost in congestion or health dangers caused by particulate matters, NO\textsubscript{x}, and stress \cite{Frakt2019}. In 2017, traffic-related externalities cost U.S. citizens 305 billion \si{\usd} \cite{INRIX2017}. 
Municipalities try to resolve these problems by improving existing transportation systems but face several obstacles \cite{HuWang2018,Hu2019}. Public transport not based on surface roads (e.g., subway lines) is often hindered by spatial limitations in urban areas and by long lead times of infrastructure projects. Improving road infrastructure often faces similar obstacles. Accordingly, better utilization of urban infrastructure using new technologies and mobility concepts is necessary to resolve the root cause of problems in today's transportation systems.

In recent years, various new mobility concepts emerged. However, all of them struggle with specific obstacles. Car-sharing systems such as Zipcar offer the opportunity to reduce the number of individually owned cars in cities. However, given the small fleet sizes, customers are reluctant to use such vehicles due to inconvenient accessibility; vice versa, these concepts are often still not economically viable for \glspl{abk:msp}. Ride-hailing services such as Uber or Lyft appear as an affordable alternative to traditional taxi services and decrease the need for individually owned cars in cities. However, the uncontrolled growth of such services can worsen congestion; e.g., in Manhattan, an increase of \SI{68000}{} for-hire vehicles from 2013 to 2018 correlated with a decrease in average travel speed from \SI{6.2}{\mph} to \SI{4.7}{\mph} \cite{Hu2019}. Ride-pooling services may reduce traffic in cities by pairing passengers with similar trips into a vehicle \cite{Alonso-Mora2017,Ostrovsky2019}. However, such systems only work efficiently when the ride-hailing system itself is efficient. So far, customers appear to be reluctant to use ride-pooling due to unsatisfactory user experiences \cite{McKinsey2017}.

\glsunset{abk:amod}
An emerging mobility concept, namely autonomous mobility-on-demand (\acrshort{abk:amod}), has the potential to help address the challenges outlined above. An \gls{abk:amod} system consists of a fleet of robotic self-driving cars, which transport passengers from their origins to their destinations.
An operator controls the fleet by dispatching passenger trips to vehicles while simultaneously deciding on their routes. Such a system overcomes the limitations of free-floating car-sharing systems by remedying the limited accessibility of cars. Further, it overcomes ride-hailing systems' inefficiencies as a single operator centrally decides on the fleet's size and all operations. It also eliminates the largest cost component of ride-hailing systems: driver's time. Central coordination and information transparency also allow for more efficient ride-pooling with fair compensation schemes among customers.

While \gls{abk:amod} has the potential to address many challenges of urban transportation, it also introduces a major risk: it may significantly impact and possibly cannibalize public transportation such as buses, trams, and subways~\cite{Oh2020}. While such cannibalization effects have already been discussed for similar modes of travel, e.g., ride-hailing services \cite{Hall2018,Jin2019}, an \gls{abk:amod} system bears a unique feature: for the first time, the centralized, fleet-wide control and coordination of vehicles allows one to employ a perfect pricing strategy; i.e., an operator can tailor her pricing strategy to her individual objective, e.g., profit maximization, without any disturbances usually resulting from driver behavior and vehicle availability. The net effect of such a system and its pricing policy on cities is mostly an empirical question, whose answer depends on many specifics of various cities. A principled approach to answer this question entails developing a unifying framework that incorporates various transportation modes to study customers' travel behavior and choices. In this paper, we provide such a framework, which makes it possible to analyze the interactions among \glspl{abk:msp}, \glspl{abk:pta} controlled by municipalities, and customers in today's and future transportation systems. In particular, we tailor this framework to the specific case of \gls{abk:amod} systems and base it on a game-theoretic perspective.

\subsubsection*{Related Literature}\label{subsec:literature}
Our work lies at the interface between economics, game theory, and transportation science. To keep our literature review concise, we focus on game-theoretic approaches in mobility systems and transportation network modeling specifically related to \gls{abk:amod}.

\textit{The traffic assignment problem} is concerned with allocating users to means of transport or links in transportation networks so that an equilibrium configuration is reached; see~\cite{Cantarella1997,Patriksson2015} for a comprehensive presentation. Differently from our work, the traffic assignment problem does not focus on the interplay among \glspl{abk:msp} but on traffic equilibria among customers.

\textit{Network pricing problems} arise in the field of traffic management and congestion avoidance \cite{Brotcorne2001,Bianco2016,Kuiteing2017}. These problems are typically modeled as Stackelberg games and formalized as mathematical programs with equilibrium constraints \cite{Patriksson2002} or as bilevel optimization problems \cite{Colson2005,Labbe2016}.
In these games, an upper-level player sets prices or tolls on some arcs in a network to maximize a given objective, and the lower-level players react accordingly.
Compared with our problem setting, such games do not accommodate the non-cooperative interaction among \glspl{abk:msp}. Moreover, they focus on arc or path prices and do not directly allow for more general pricing schemes.

\textit{Mobility-on-demand related games} have recently been extensively discussed. \cite{Banerjee2015,Bai2019,Guda2019} focused on the coordination of customer demand and driver supply. Moreover, \cite{Bimpikis2019,Wang2018} studied ride-sharing platforms, highlighting the impact of the demand pattern on the platforms' profits and consumers' surplus, while studying cost-sharing strategies between customers and drivers.
Further works have focused on the societal costs of ride-hailing companies \cite{Rogers2015} and on the impact of mobility-on-demand systems on the taxi market \cite{Wallsten2015}.
Overall, these approaches do not sufficiently capture our problem characteristics as they (\textit{i})~focus on a two-sided market with drivers and customers, not accounting for centrally-controlled autonomous vehicles, (\textit{ii})~do not consider multimodal or intermodal routes, and (\textit{iii})~do not provide a general and flexible game-theoretic framework that captures both the interactions among \glspl{abk:msp} as well as among \glspl{abk:msp} and customers. 

Focusing on \textit{transportation models for \gls{abk:amod} systems}, previous papers have investigated queuing-theoretic models, simulation-based models, and 
network flow models. We refer to \cite{ZardiniLanzettiAnnRev2021,Narayanan2020} for extensive reviews of this field.
Microscopic studies expect autonomous vehicles to ease traffic management, e.g., via improved intersection clearing \cite{Lee2012,Guler2014} and freeway merging \cite{Zhou2017}. Macroscopic studies have shown that \gls{abk:amod} systems contribute to more accessible, efficient, and sustainable transportation systems \cite{Fagnant2015}, but their unregulated deployment could worsen congestion, cannibalize public transportation, and increase vehicle-kilometers traveled~\cite{Oh2020}.
Further, researchers studied the joint design and operation of \gls{abk:amod} systems and the public transportation system~\cite{SalazarLanzettiEtAl2019,Banerjee2021,ZardiniLanzettiTNSE2022} and the economics of \gls{abk:amod} systems, focusing on the effects of carpooling \cite{Ostrovsky2019}. All approaches but \cite{Ostrovsky2019} imposed a central decision maker and neglected game-theoretic dynamics, while \cite{Ostrovsky2019} neglected \gls{abk:msp} behavior, essentially replacing \glspl{abk:msp} with a perfectly competitive, zero-profit market. In contrast, one of the key elements in our work is the consideration of \glspl{abk:msp}' behavior.

Concluding, to the best of our knowledge, there exists no methodological framework capable of analyzing the dynamics among customers and multiple \glspl{abk:msp} offering different mobility services while considering the operational constraints of the underlying transportation system.

\subsubsection*{Contribution}\label{subsec:scope}
To fill this research gap, we provide the first game-theoretic framework that captures the dynamics among multiple \glspl{abk:msp} and customers while accounting for the system's operational constraints. Specifically, our scientific contribution is fourfold.
First, we develop a generic mathematical framework that allows us to analyze the dynamics of complex transportation problems by combining graph-theoretic network models with game-theoretic approaches that consider the interconnections between a transportation network and its corresponding market place.
Second, we tailor this framework to the specific cases of two \gls{abk:amod} operators interacting between themselves and with public transport as well as customers. Moreover, we consider the degenerate case of a monopolistic \gls{abk:amod} operator and study the equilibria of the arising games. 
Third, we provide a real-world case study for the city of Berlin.
Fourth, we present extensive numerical results and sensitivity analyses, yielding managerial findings on the interaction between \gls{abk:amod} operators and public transport. Among others, we show that the system equilibria yield similar modal shares for both the degenerate monopolistic and the competitive cases, while we observe significant differences in the operators' profit.
\subsubsection*{Organization}\label{subsec:organization}
The remainder of this paper is as follows. We specify our problem setting in \cref{sec:problemsetting} and develop our methodology in \cref{sec:methodology}. In \cref{sec:amodgame}, we tailor our methodology to study the interplay between \gls{abk:amod} operators and a public transportation system. \cref{sec:case study} details our experimental design, focusing on a real-world case study. We present results in \cref{sec:results}. \cref{sec:conclusion} concludes this paper by summarizing its main findings.
In the appendices, we provide fundamentals of graph theory (Appendix~\ref{app:graphtheory}) and proofs for all propositions (Appendix~\ref{app:proofs}). When introducing a term defined in the appendix for the first time, we mark it with a dagger\myast{}.

%% file: chapters/2_problemsetting.tex
\section{Problem Setting}\label{sec:problemsetting}
In this work, we focus on intra-city passenger transportation, where an \gls{abk:amod} fleet substitutes the service of current taxi or ride-hailing fleets. In such a system, different \glspl{abk:msp} interact with each other and with customers. We distinguish \glspl{abk:msp} between commercial \glspl{abk:msp} and municipalities and focus on three stakeholder groups.

\textit{Mobility service providers} offer transportation services to customers in order to maximize their profit. \glspl{abk:msp} require cost-effective operations, ignoring the resulting externalities and their effect on the overall system or on other players. Accordingly, an \gls{abk:msp} aims to influence customer decisions towards her own benefit via dedicated pricing strategies. The effect of such pricing strategies reaches a maximum for an \gls{abk:amod} fleet, where no intermediaries (e.g., drivers) may perturb the operator's strategy, and the operator remains in full control of the fleet behavior.

\textit{Municipalities} offer transportation services via a \gls{abk:pta} to customers while aiming to sustain infrastructure services, accessibility, and quality of life. While \glspl{abk:msp} complement the transportation services offered by municipalities, they also cause externalities and dissatisfaction. Accordingly, municipalities try to influence \glspl{abk:msp}, e.g., through subsidies or taxes. 

\textit{Customers} represent the demand side and request for transportation services. Customers can choose between different transportation modes or combine them to complete their ride. Each customer selects a trip in line with her preferences, e.g., minimizing her cost, travel time, or a combination of both.

To adequately capture such a system's dynamics and the interactions between stakeholders, we model the city's transportation system on two different levels: a transportation network and its market place (see \cref{fig:transmarket}). 

\textit{Transportation market place:} The interaction between the different stakeholders occurs in the system's \textit{market place}, e.g., via a smartphone app. Here, \glspl{abk:msp} and municipalities offer several transportation services to customers. Customers have different transportation demands and respond to these offers depending on their individual rationale. These interactions happen at the operational level in a short time horizon. At the strategic level, \glspl{abk:msp} interact and therefore influence each other as their business models may interfere, i.e., a customer may substitute the service of one provider with that of another provider, depending on quality and price. In this setting, an \gls{abk:amod} operator has a specific advantage over other operators. Compared to a \gls{abk:pta}, the \gls{abk:amod} operator can leverage short-term pricing policies to influence customer behavior. While a regular ride-hailing operator may utilize short-term pricing policies as well, an \gls{abk:amod} operator is in full control of the fleet and does not face disturbances from intermediaries such as drivers. This makes a short-term pricing policy more effective and usable as a surgical tool to steer the system. Assessing the impact of this advantage is a central motivation for this work.
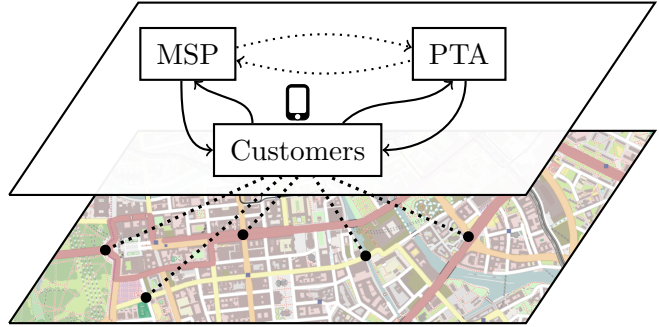
\begin{figure}[!t]
	\centering
	\input{graphics/tikz/gameSketch.tex}
	\caption{Schematic model of a transportation system with one commercial \gls{abk:msp} and a \gls{abk:pta}. The upper level is the transportation market place, where the different stakeholders interact. The lower level is the transportation network, where the realization of a demand and supply match takes place.}    
	\label{fig:transmarket}
\end{figure}

\textit{Transportation network:} The realization of a demand and supply match between customers and mobility service providers occurs in the transportation network, consisting of the city's road and public transportation networks. Accordingly, the transportation network imposes boundaries on offers in the transportation market place as it determines the available infrastructure and comprises externalities such as congestion and travel times.

In this paper, we focus specifically on the interaction between \gls{abk:amod} operators and a municipality, offering transportation service through a \gls{abk:pta}. Here, we focus on a short-term perspective, i.e., the equilibrium for a snapshot of a day. Given this time horizon, the \gls{abk:pta} does not change her prices at such an operational level because public transport tariffs are set strategically for a significantly longer time horizon. Conversely, \gls{abk:amod} operators change their prices at the operational level to maximize their profit. When taking the pricing decision, an \gls{abk:amod} operator considers potential profits from serving customers and additional costs resulting from relocating vehicles after finished trips. To this end, rebalancing empty vehicles can be interpreted as a reorganization of vehicle positions to match anticipated demand.

%% file: graphics/tikz/gameSketch.tex
\tikzstyle{node graph} = [circle,thick]
\tikzstyle{subway line} = [-,thick,color=yellow]
\tikzstyle{road} = [-,ultra thick]
\tikzstyle{walk} = [-,thick,color=blue]

\def\deltax{0.0}

\begin{tikzpicture}[scale=0.60]
    \draw[thick,fill overzoom image=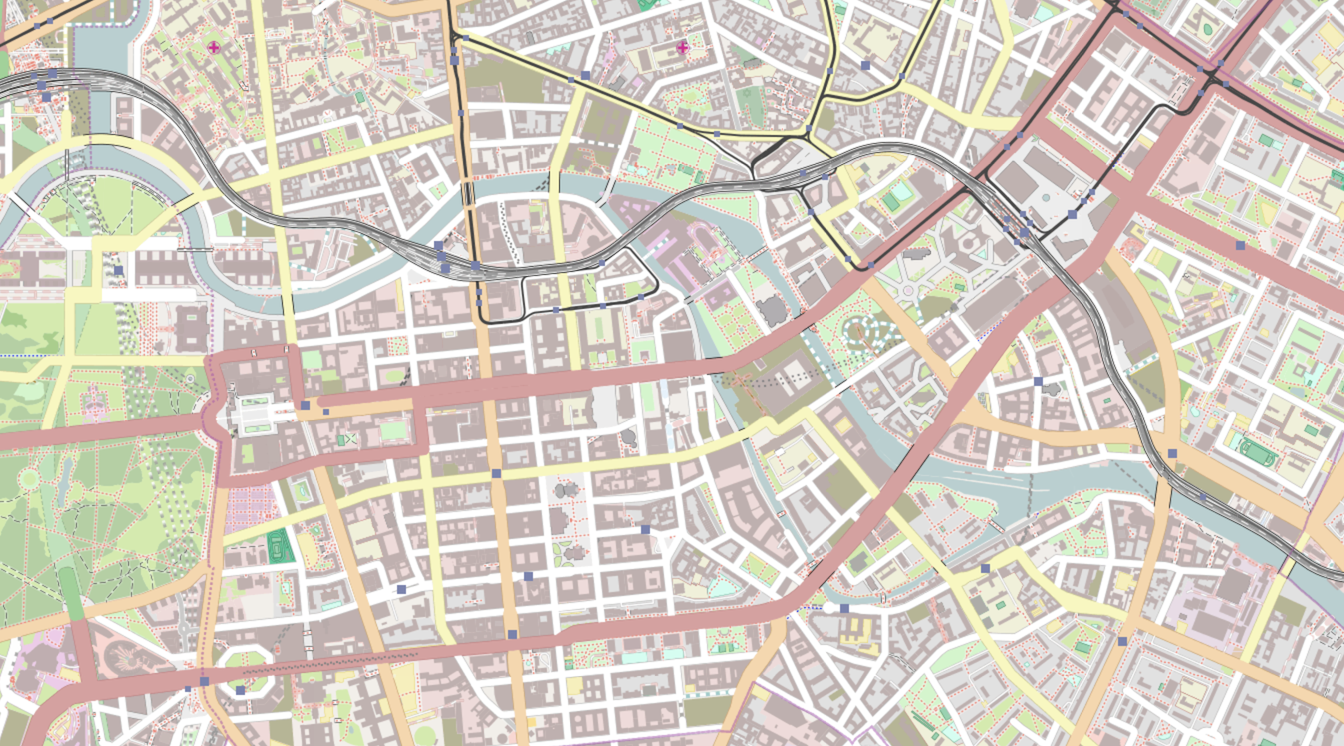] (-4,-2) --++ (8,0) --++ (2,3) --++ (-8,0) --++ (-2,-3);
    
    
    \draw[thick,fill=white,fill opacity=0.95] ($(-4,-2)+(0,2)$) --++ (8,0) --++ (2,3) --++ (-8,0) --++ (-2,-3); 
    
    \node[rectangle,draw,fill=white,thick,inner sep=6pt] at (0.7*2/3,0.7) (customers) {Customers}; 
    \node[rectangle,draw,fill=white,thick,inner sep=6pt] at (2.2*2/3+2-0.5,2.2) (mcp) {PTA};  
    \node[rectangle,draw,fill=white,thick,inner sep=6pt] at (2.2*2/3-2-0.7,2.2) (msp) {MSP};  
    
    \draw[->,thick] ($(customers.north)+(+0.7,0)$) to[out=45,in=-135] ($(mcp.south)+(-0.1,0)$);
    \draw[->,thick] ($(customers.north)+(-0.7,0)$) to[out=90,in=-45] ($(msp.south)+(0.1,0)$);
    \draw[->,thick] ($(msp.south)+(-0.1,0)$) to[out=-90,in=180] (customers.west);
    \draw[->,thick] ($(mcp.south)+(+0.1,0)$) to[out=-90,in=0] (customers.east);
    
    \draw[->,dotted,thick] ($(msp.east)+(0,0.1)$) to[out=15,in=165] ($(mcp.west)+(0,0.1)$);
    \draw[<-,dotted,thick] ($(msp.east)+(0,-0.1)$) to[out=-15,in=-165] ($(mcp.west)+(0,-0.1)$);

    \draw[-,dotted,very thick] (-2.51,-0.86) -- ($(customers.south)+(-0.5,0)$) node[pos=0,circle,draw,solid,minimum size=0.1cm,fill=black,inner sep=0pt] {};
    \draw[-,dotted,very thick] (-1.88,-1.60) -- ($(customers.south)+(-0.25,0)$) node[pos=0,circle,draw,solid,minimum size=0.1cm,fill=black,inner sep=0pt] {};
    \draw[-,dotted,very thick] (-0.38,-0.62) -- ($(customers.south)+(0,0)$) node[pos=0,circle,draw,solid,minimum size=0.1cm,fill=black,inner sep=0pt] {};
    \draw[-,dotted,very thick] (1.52,-0.95) -- ($(customers.south)+(0.25,0)$) node[pos=0,circle,draw,solid,minimum size=0.1cm,fill=black,inner sep=0pt] {};
    \draw[-,dotted,very thick] (3.11,-0.65) -- ($(customers.south)+(0.5,0)$) node[pos=0,circle,draw,solid,minimum size=0.1cm,fill=black,inner sep=0pt] {};
    
    \def\phoneWidthPlot{0.35}
    \def\phoneHeightPlot{0.55}
	\node at ($(customers.north)+(-0.5*\phoneWidthPlot,0.05)$) (phone) {};
    \draw[rounded corners=2pt,fill=black] (phone) rectangle ++(\phoneWidthPlot,\phoneHeightPlot);
    \draw[rounded corners=1pt,fill=white] ($(phone)+(+0.1*\phoneWidthPlot,0.1)$) rectangle ++(0.8*\phoneWidthPlot,0.75*\phoneHeightPlot);
    \draw[rounded corners=2pt,fill=white] ($(phone)+(+0.5*\phoneWidthPlot,0.05)$) circle (0.05cm);
     
\end{tikzpicture}

%% file: chapters/3_methodology.tex
\section{Methodology}\label{sec:methodology}
We now develop the methodology to analyze the problem setting introduced in \cref{sec:problemsetting}. We first show in \cref{subsec:graphsetting} how to formalize the transportation network via graph theory. We then introduce a fundamental game-theoretic model for a transportation market place in \cref{subsec:gametheorybasics}.
\subsection{Graph Representation of a Multi-stakeholder Transportation System}
\label{subsec:graphsetting}
We represent a transportation network on a multigraph\myast{} $\graphDefinition{}$ with a vertex set $\verticesSet{}$, an arc set $\arcsSet{}$, and identifiers $\arcSource{}:\arcsSet{}\to\verticesSet{}$ and $\arcTarget{}:\arcsSet{}\to\verticesSet{}$ assigning each arc to its source and sink.
Each vertex $\vertex{}\in\verticesSet{}$ denotes a location where a customer can start or end her trip. Each arc $\arc{}\in\arcsSet{}$ represents a certain transportation mode for a trip, e.g., a self-driving car or a subway line. Accordingly, multiple arcs may exist between any two vertices $\vertex{1},\vertex{2}$ to model the available modes of transportation. 
We define arc subsets $\arcsSet{j}\subset\arcsSet{}$, inducing subgraphs $\graphDefinition{j}$ with $\verticesSet{j}\coloneqq\bigcup_{\arc{}\in\arcsSet{j}}\{\arcSource{}(\arc{}),\arcTarget{}(\arc{})\}$, $\arcSource{j}\coloneqq\arcSource{}|_{\arcsSet{j}}$, and $\arcTarget{j}\coloneqq\arcTarget{}|_{\arcsSet{j}}$ (i.e., the restrictions of  $\arcSource{}$ and $\arcTarget{}$  to $\arcsSet{j}$). Each subgraph denotes a homogeneous mode of transportation, e.g., the subway network or the service region of an \gls{abk:amod} system.
Additionally, each subgraph $\graph{j}$ is controlled by an \gls{abk:msp} (e.g., through price setting), from here on referred to as its operator.
The subgraph $\graphDefinition{0}$ denotes a part of the overall network where customers can move free of (monetary) charge, e.g., by walking.
Accordingly, neither an operator nor prices exist for the subgraph $\graph{0}$.

To reflect customers' choice options in a real transportation system, we assume $\graph{0}$ to be non-trivial (i.e., $\verticesSet{0}\neq \emptyset$) and strongly connected\myast{}. Further, the subgraphs' arc sets are disjoint: formally $\arcsSet{j}\cap\arcsSet{l}=\emptyset,\ \forall j\in\{0,1,\ldots,\operatorsNumber\}, l\in\{0,1,\ldots,\operatorsNumber\}\setminus\{j\}$, where $\operatorsNumber$ denotes the number of operators in the system.
While these properties allow customers to choose between all available transportation modes, they do not prevent various operators from providing service on the same road, as multigraphs allow multiple arcs for an origin-destination pair. 

We use a time-invariant network flow representation for customer movements and differentiate between transportation requests and transportation demands. A request refers to a single customer~$k$ and is well-defined by a pair $(\originVertex{k},\destinationVertex{k})\in\verticesSet{}\times\verticesSet{}$, which states the origin~$\originVertex{k}$ and the destination~$\destinationVertex{k}$ of her trip. Conversely, a demand\myast{} $\demand{i}$ aggregates identical requests from different customers via origin, destination, and demand rate; it therefore refers to a customer flow. 

\begin{definition}[Demand]
	\label{definition:demand}
	A demand $\demand{i}$ is a triple $(\originVertex{i},\destinationVertex{i},\demandRate{i})\in\verticesSet{}\times\verticesSet{}\times\positiveReals$ uniquely defined by its origin $\originVertex{i}$, its destination $\destinationVertex{i}$, and a demand rate $\demandRate{i}$, which results from all requests coinciding in $\originVertex{i}$ and $\destinationVertex{i}$.
	For an arbitrary set of $M$ demands with label set $\{1,\ldots,\demandNumber\}$, we define $\demandsSet{}\coloneqq\{\demand{1},\ldots,\demand{\demandNumber}\}$ as the set of all demands.
\end{definition}

We stipulate that each demand starts and ends on the subgraph $\graph{0}$; else, if a trip starts or ends on the subgraph of \gls{abk:msp} $j$, then \gls{abk:msp} $j$ can pick arbitrarily large prices and the customer is forced with using the transportation service.
Formally, for all $\demand{i}=(\originVertex{i},\destinationVertex{i},\demandRate{i})\in\demandsSet{}$ we have $\originVertex{i},\destinationVertex{i}\in\verticesSet{0}$. 
In general, a given demand may be satisfied by multiple paths\myast. With $\pathsSet{\arcsSet{}}$ being the set of all paths in $\graphDefinition{}$, we can define the demand-satisfying paths.
\begin{definition}[Demand-satisfying paths]
	\label{definition:pathsatreq}
	A set of $L$ paths $\{\path{1},\ldots,\path{L}\}\subseteq \pathsSet{\arcsSet{}}$ satisfies a demand $\demand{i}=(\originVertex{i},\destinationVertex{i},\demandRate{i})$ if the origins\myast~ $\originPath{}(\path{j})$ and destinations\myast~ $\destinationPath{}(\path{j})$ of the paths and of the demand coincide, i.e., if $\originPath{}(\path{j})=\originVertex{i}$ and $\destinationPath{}(\path{j})=\destinationVertex{i}$ for all $j\in\{1,\ldots,L\}$.
	We denote by $\pathsSetRequest{\demand{i}}\subseteq \pathsSet{\arcsSet{}}$ the set of all paths satisfying demand $\demand{i}$.
\end{definition}
Then, according to the customers' choices, the demand induces flows\myast{}, defined by a path and a flow rate. We denote the path and the rate corresponding to a flow $\flow{}$ by $\projectionPathFlow{\flow{}}$ and $\projectionRateFlow{\flow{}}$.

\subsection{Game-theoretic Setting}
\label{subsec:gametheorybasics}
We focus on the interplay among (different) \glspl{abk:msp} and customers: 
first, customers set requests in $\graph{0}$.
Second, the \glspl{abk:msp} decide on prices for their offered services on their subgraphs $\graph{j}$.
Third, customers choose a transportation service which then results in transport flows on~$\graph{}$.
We illustrate the game for the simplified case of one \gls{abk:amod} \gls{abk:msp} and a \gls{abk:pta} controlled by a municipality in \cref{fig:multigraph customers}.
In the following, we formalize the customer and operator decisions and define game equilibria.

\begin{figure}[!b]
    \newif\ifshowprices
    \newif\ifshowpaths
    \newif\ifshowflows
    \newif\ifshowrebalancing
    \def\distanceBlockVert{1.6}
	\def\distanceBlockHor{1.6}
    \def\angle{65}
    \centering
    \begin{subfigure}[t]{0.32\columnwidth}
        \showpricestrue
        \centering
        \input{graphics/tikz/multigraphSketchCustomers.tex}
        \showpricesfalse
        \caption{ \footnotesize
        Example of (some of the) prices set by an \gls{abk:amod} operator for the origin-destination pairs on her subgraph (the road subgraph, depicted in blue). Customers see these prices and react accordingly.} 
    \end{subfigure}
    \hfill
    \begin{subfigure}[t]{0.32\columnwidth}
        \showpathstrue
        \centering
        \input{graphics/tikz/multigraphSketchCustomers.tex}
        \showpathsfalse
        \caption{ \footnotesize
        Example of demand-satisfying paths (i.e., possible customer reactions) for the demand with the black vertex as an origin and the grey vertex as a destination. The upper path combines public transport (red) and walking (yellow), whereas the lower path combines an \gls{abk:amod} (blue) ride and walking (yellow).} 
    \end{subfigure}
    \hfill 
    \begin{subfigure}[t]{0.32\columnwidth}
        \showpathstrue
        \showflowstrue
        \showrebalancingtrue
        \centering
        \input{graphics/tikz/multigraphSketchCustomers.tex}
        \showflowsfalse
        \showrebalancingfalse
        \showpathsfalse
        \caption{ \footnotesize
        Assuming some customers choose the \gls{abk:amod} ride, the \gls{abk:amod} \gls{abk:msp} will deploy vehicle flows to serve the demand and to rebalance her fleet to ensure vehicle conservation at each node. The \gls{abk:amod} \gls{abk:msp} only serves the mobility demand induced on her subgraph (e.g., only the blue arc in the lower path).}
    \end{subfigure}
    \caption{Schematic example of the different game stages: (a) the \gls{abk:amod} \gls{abk:msp} sets the prices for all origin-destination pairs in her graph; (b) each demand splits over available paths; (c) the \gls{abk:amod} \gls{abk:msp} operator serves her demand share and rebalances the fleet.}
    \label{fig:multigraph customers}
\end{figure}
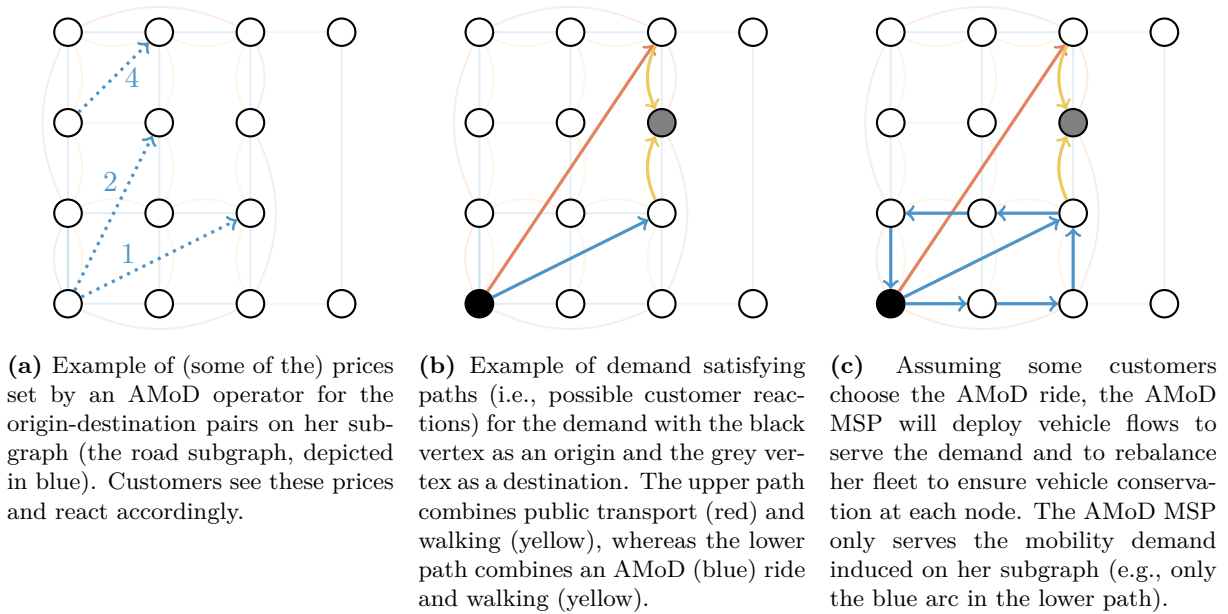

\subsubsection{Customers' Reactions and Operators' Decisions}
Formally, we focus on a simultaneous game between $\operatorsNumber$ operators.
Herein, customers act as non-strategic reactive players.
\paragraph{Customer reaction} 
Customers may travel free of charge by using arcs in the subgraph $\graph{0}$ (i.e., they walk) or may request an \gls{abk:msp}'s service. If they request such a service, they cannot influence the \gls{abk:msp}'s operations. Instead, they only ask for a mobility service between an origin and a destination. Formally, customers select arcs in the \emph{fully-connected}\myast{} operators' subgraph denoted by $\graphDefinition[bar]{j}$.
More specifically, customers select paths in the graphs' arc sets $\arcsSet{0}\cup\bigcup_{j=1}^{\operatorsNumber}\arcsSet[bar]{j}$. As customers may choose different services, each demand $\demand{i}$ splits over its demand-satisfying paths $\pathsSetRequest{\demand{i}}\subseteq\pathsSet{\arcsSet{0}\cup\bigcup_{j=1}^{\operatorsNumber}\arcsSet[bar]{j}}$.
To model this customer behavior, we use a reaction curve \mbox{$\reactionCurve{i}:\pathsSetRequest{\demand{i}}\to[0,\alpha_i]$} for each demand $\demand{i}=(\originVertex{i},\destinationVertex{i},\demandRate{i})$ which assigns a share of the total demand to each demand-satisfying path.
We introduce the customers' action space $\customersActionSpace{i}$ to capture (\textit{i}) demand conservation, i.e., $\sum_{\path{}\in\pathsSetRequest{\demand{i}}}\reactionCurve{i}(\path{})=\demandRate{i}$,
and (\textit{ii}) additional constraints customers might face.
For instance, if demand $\demand{i}$ is not allowed to travel through the path $\path[bar]{}\in\pathsSetRequest{\demand{i}}$, then the corresponding action space reads $\customersActionSpace{i}=\{\reactionCurve{}\in\nonnegativeReals^{\pathsSetRequest{\demand{i}}}\,|\, \reactionCurve{}(\path[bar]{})=0, \sum_{\path{}\in\pathsSetRequest{\demand{i}}}
	\reactionCurve{i}(\path{})=\demandRate{i}\}$.
Finally, each reaction curve is associated with a non-negative cost $\demandCost{i}:\customersActionSpace{i}\times\priceStrategiesSet{1}\times\ldots\times\priceStrategiesSet{\operatorsNumber}\to\nonnegativeReals$, with $\priceStrategiesSet{j}$ denoting the set of pricing strategies available to operator $\operatorsObject{j}$ as elaborated below.
For instance, this cost could represent the sum of the fares paid throughout the trip or the total cost, consisting of the fares paid and a customer's monetary value of time.
\paragraph{Operator decision} Each operator $\operatorsObject{j}$ serves customer demands, herein operating her fleet to maximize profit.
Her profit depends on the customer fares and the operational costs, resulting from customer-transporting vehicles and operating empty vehicles, which the operator relocates in the system to balance her flows.
Recall that customers move on the \emph{fully-connected} version of the operators' subgraph and on the non-controlled graph $\graph{0}$. 
Accordingly, an arc on the operator's fully-connected subgraph denotes the customers' possibility to request a ride from a source vertex to a sink vertex. A ride is only realized if the arc lies on a demand-satisfying path $\path{}\in\pathsSetRequest{\demand{i}}$ with $\reactionCurve{i}(\path{})>0$ for some demand $\demand{i}\in\demandsSet{}$.
To satisfy these requests on her subgraph $\graph{j}$, operator $\operatorsObject{j}$ selects vehicle flows: for each requested ride, there must be a vehicle flow (\textit{i}) whose origin and destination match the ones of the ride and (\textit{ii}) whose rate equals the rate of customers who requested the ride. Further, we require that each vehicle flow does not simultaneously serve two rides. We call such a set of flows potentially active.
\begin{definition}[Potentially active set of flows]
\label{definition:potentially active set of flows}
Consider a demand $\demand{i}$ with reaction curve $\reactionCurve{i}$. A set of flows $\flowsSet{i}$ on operator $\operatorsObject{j}$'s subgraph $\graph{j}$ is potentially active for demand $\demand{i}$ if for each demand-satisfying path $\path{}\in\pathsSetRequest{\demand{i}}$ with $\reactionCurve{i}(\path{})>0$, and each arc $\arc{}\in\path{}\cap\arcsSet[bar]{j}$ there are flows $\flow{i}^{a,p}\in\flowsSet{i}$ which satisfy the mobility service induced by demand $\demand{i}$ such that 
(\textit{i}) its origin and destination coincide with $\arcSource[bar]{j}(\arc{})$ and $\arcTarget[bar]{j}(\arc{})$ and (\textit{ii}) its rate matches the demand share on path $\path{}$, i.e., $\projectionRateFlow{\flow{i}^{a,p}}=\reactionCurve{i}(\path{})$.
To prevent a flow from serving multiple rides, we require that the flows $\flow{i}^{a,p}$ do not coincide; else, a flow $\flow{}=\flow{i}^{a_1,p_1}=\flow{i}^{a_2,p_2}$ would simultaneously serve ride $\arc{1}$ on path $\path{1}$ and ride $\arc{2}$ on path $\path{2}$. We denote the set of potentially active flow sets by $\flowsSetDemand{i}(\reactionCurve{i})$.
\end{definition}

The operator may additionally select a set of flows of rebalancing vehicles $\flowsSet{0}$. 
To ensure that flows are balanced and that additional constraints (e.g., limited vehicles availability) are fulfilled, we introduce the operator's action set $\operatorsActionSpace{j}\subseteq 2^{\flowsSetGraph{\graph{j}}}\times\ldots\times 2^{\flowsSetGraph{\graph{j}}}$ and impose $(\flowsSet{1},\ldots,\flowsSet{M},\flowsSet{0})\in\operatorsActionSpace{j}$.
In this setting, the revenue of an operator reads
\begin{equation}\label{eq:operator revenue}
	\summe{i\in\demandIndexSet,\path{}\in \pathsSetRequest{\demand{i}},a\in \path{}\cap\arcsSet[bar]{j}}\qquad 
	\reactionCurve{i}(\path{})\cdot \priceStrategy{j}(\arcSource[bar]{j}(\arc{}),\arcTarget[bar]{j}(\arc{})),
\end{equation}
depending on the prices set by the operator, where we refer to a representation of all prices set by the operator as a pricing strategy $\priceStrategy{}$.
\begin{definition}[Pricing strategy]
	\label{definition:pricing strategy}
	Consider a multigraph $\graphDefinition{j}$.
	A pricing strategy $\priceStrategy{}:\verticesSet{j}\times\verticesSet{j}\to\nonnegativeReals$ on $\graph{j}$ is a function assigning a non-negative price $\priceStrategy{}(o,d)$ to each origin-destination pair $(o,d)\in\verticesSet{j}\times\verticesSet{j}$.
\end{definition}
As multiple pricing strategies can exist, we collect the feasible pricing strategies on a graph $\graphDefinition{j}$ in the set $\priceStrategiesSet{j}\subseteq \nonnegativeReals^{\verticesSet{j}\times\verticesSet{j}}$. In line with realistic constraints on pricing strategies, we assume $\priceStrategiesSet{j}$ to be closed, convex, and of the form
$\priceStrategiesSet{j}
    = \{\left.\priceStrategy{}\in\nonnegativeReals^{\verticesSet{j}\times\verticesSet{j}}\,\right|\,g_j(\priceStrategy{}(o,d),o,d)\leq 0\;\forall o, d\in\verticesSet{j}\}$
for some $g_j:\nonnegativeReals\times\verticesSet{j}\times\verticesSet{j}\to\reals^l$, with $l\in\mathbb{N}$, expressing the constraints on the price $\priceStrategy{}(o,d)$; e.g., if prices are regulated by an upper bound $B\in\nonnegativeReals$, then $g_j(\priceStrategy{j}(o,d),o,d)=\priceStrategy{j}(o,d)-B$.

To model the operator's costs (e.g., energy consumption, maintenance, and depreciation of the vehicles), we introduce a cost $\operationCostFunction{j}(a)$ for traversing arc $\arc{}\in\arcsSet{j}$, resulting in the cost $\projectionRateFlow{\flow{}}\sum_{\arc{}\in\projectionPathFlow{\flow{}}}\operationCostFunction{j}(\arc{})$ for a flow $\flow{}$. Thus, the total cost is
\begin{equation}\label{eq:operator cost}
    \!\!\!\!\!
    \min_{\substack{\flowsSet{i}\in\flowsSetDemand{i}(\reactionCurve{i}), \flowsSet{0}\in 2^{\flowsSetGraph{\graph{j}}}, \\
        		(\{\flowsSet{i}\}_{i=1}^{\demandNumber},\flowsSet{0})\in\operatorsActionSpace{j}}}\,
        \sum_{i\in\{0,1,\ldots,\demandNumber\},\flow{}\in\flowsSet{i}} 
        \projectionRateFlow{\flow{}}
        \sum_{\arc{}\in\projectionPathFlow{\flow{}}}
        \operationCostFunction{j}(\arc{}).
\end{equation}
Finally, the operator's profit $\utility_j:\priceStrategiesSet{j}\times\prod_{i\in\demandIndexSet}\customersActionSpace{i}\to\reals$, depending on the operator's pricing strategy and on the customers' reaction curves, results from the excess of revenue (Eq. \eqref{eq:operator revenue}) over costs (Eq. \eqref{eq:operator cost}). 

A few comments on this general setting are in order.
First, we do not include direct interactions among customers, but our model can be extended to accommodate them. In particular, one can define customers as strategic players, interacting simultaneously with themselves and sequentially with \glspl{abk:msp}. In line with this, we leave effects such as ride pooling to future research. 
Second, we assume without loss of generality that an operator serves all customer requests. 
However, an operator can technically drop a customer by imposing an artificially high transportation fare, which causes the customer to refuse to choose a ride.
Third, we neglect the operator's fixed costs because they do not affect the operational decisions.
Fourth, we consider a time-invariant transportation system. This assumption reflects the mesoscopic nature of our study, i.e., we consider a representative snapshot of the transportation system to account for a realistic behavior of \glspl{abk:msp} and customers within the studied game.
Fifth, we assume that a customer always decides to travel. One may easily relax this assumption by adding an auxiliary subgraph representing customers who refrain from traveling.

\subsubsection{Game Equilibrium}
As a basis for the definition of game equilibria, we first introduce a customer demand's optimal reaction: A demand reacts optimally if its reaction curve minimizes its cost for given operators' pricing strategies.
\begin{definition}[Optimal reaction]
	\label{definition:customer equilibrium}
	Given the operators' pricing strategies $(\{\priceStrategy{j}\}_{j=1}^{\operatorsNumber})\in\prod_{j=1}^{\operatorsNumber}\priceStrategiesSet{j}$, the reaction curve $\reactionCurve[star]{i}\in\customersActionSpace{i}$ is optimal for demand $\demand{i}$ if $\demandCost{i}(\reactionCurve[star]{i},\{\priceStrategy{j}\}_{j=1}^{\operatorsNumber})
    	\leq 
    	\demandCost{i}(\reactionCurve{i},\{\priceStrategy{j}\}_{j=1}^{\operatorsNumber})$
    for all $\reactionCurve{i}\in\customersActionSpace{i}$.
	Let $\customersEquilibria{i}(\{\priceStrategy{j}\}_{j=1}^{\operatorsNumber})$ be the set of all optimal reactions.
\end{definition}
Since operators interact simultaneously, we say that the game is at equilibrium if none of the operators can increase her profit by unilaterally changing her pricing strategy, given that customers react optimally.
To keep the definition of such an equilibrium concise, we assume without loss of generality that for all operators' pricing strategies $(\{\priceStrategy{j}\}_{j=1}^{\operatorsNumber})\in\prod_{j=1}^{\operatorsNumber}\priceStrategiesSet{j}$ and all demands $i\in\demandIndexSet$, the set $\customersEquilibria{i}(\{\priceStrategy{j}\}_{j=1}^{\operatorsNumber})$ is a singleton which ensures profit uniqueness for a pricing strategy. A relaxation of this assumption is however straightforward; e.g., by introducing a selection function for the sets $\customersEquilibria{i}(\{\priceStrategy{j}\}_{j=1}^{\operatorsNumber})$.

Denoting the operator's profit as $\utility_j(\priceStrategy{j},\{\customersEquilibria{i}(\priceStrategy{j},\priceStrategy{-j})\}_{i=1}^{\demandNumber})$, where $\priceStrategy{-j}\in\prod_{k=1,k\neq j}^{\operatorsNumber}\priceStrategiesSet{k}$ is the tuple of the pricing strategies of all operators but operator $j$, we can define equilibria. 
\begin{definition}[Game equilibrium]
	\label{definition:game equilibrium}
	The pricing strategies $\equilibriumGeneral[star]\in\prod_{j=1}^{\operatorsNumber}\priceStrategiesSet{j}$ are an equilibrium of the game if no operator can increase her profit by unilaterally deviating from her pricing strategy: $\equilibriumGeneral[star]$ is an equilibrium if for all $j\in\{1,\ldots,\operatorsNumber\}$ and all $\priceStrategy{j}\in\priceStrategiesSet{j}$ $\utility_j(\priceStrategy[star]{j},\{\customersEquilibria{i}(\priceStrategy[star]{j},\priceStrategy[star]{-j})\}_{i=1}^{\demandNumber})
    	\geq
    	\utility_j(\priceStrategy{j},\{\customersEquilibria{i}(\priceStrategy{j},,\priceStrategy[star]{-j})\}_{i=1}^{\demandNumber})$.
\end{definition}

%% file: graphics/tikz/multigraphSketchCustomers.tex
\tikzstyle{node graph} = [circle,thick]
\tikzstyle{public} = [-,thick,color=publicTransit,opacity=0.15]
\tikzstyle{road} = [-,thick,color=amod,opacity=0.15]
\tikzstyle{walk} = [-,thick,color=walking,opacity=0.15]

\begin{tikzpicture}[scale=0.45]
	\ifshowprices
    \node[node graph,draw] at (0,0) (node00) {};
    \else
    \node[node graph,draw,fill=black] at (0,0) (node00) {};
    \fi
    \node[node graph,draw] at (0,\distanceBlockVert) (node01) {};
    \node[node graph,draw] at (0,2*\distanceBlockVert) (node02) {};
    \node[node graph,draw] at (0,3*\distanceBlockVert) (node03) {};
    \node[node graph,draw] at (\distanceBlockHor,0) (node10) {};
    \node[node graph,draw] at (\distanceBlockHor,\distanceBlockVert) (node11) {};
    \node[node graph,draw] at (\distanceBlockHor,2*\distanceBlockVert) (node12) {};
    \node[node graph,draw] at (\distanceBlockHor,3*\distanceBlockVert) (node13) {};
    \node[node graph,draw] at (2*\distanceBlockHor,0) (node20) {};
    \node[node graph,draw] at (2*\distanceBlockHor,\distanceBlockVert) (node21) {};
    \ifshowprices
    \node[node graph,draw] at (2*\distanceBlockHor,2*\distanceBlockVert) (node22) {};
    \else
    \node[node graph,draw,fill=black!50!white] at (2*\distanceBlockHor,2*\distanceBlockVert) (node22) {};
    \fi
    \node[node graph,draw] at (2*\distanceBlockHor,3*\distanceBlockVert) (node23) {};
    \node[node graph,draw] at (3*\distanceBlockHor,0) (node30) {};
    \node[node graph,draw] at (3*\distanceBlockHor,3*\distanceBlockVert) (node33) {};
    
    \draw[road]    (node00) -- (node01);
    \draw[road]    (node01) -- (node02);
    \draw[road]    (node02) -- (node03);
    \draw[road]    (node10) -- (node11);
    \draw[road]    (node11) -- (node12);
    \draw[road]    (node12) -- (node13);
    \draw[road]    (node20) -- (node21);
    \draw[road]    (node21) -- (node22);
    \draw[road]    (node22) -- (node23);
    
    \draw[road] (node00) -- (node01);
    \draw[road] (node01) -- (node11);
    \draw[road] (node11) -- (node21);
    \draw[road] (node10) -- (node20);
    \draw[road] (node03) -- (node13);
    \draw[road] (node13) -- (node23);
    \draw[road] (node00) -- (node10);
    
    \draw[walk] (node00) to[out=\angle,in=-\angle] (node01);
    \draw[walk] (node01) to[out=\angle,in=-\angle] (node02);
    \draw[walk] (node02) to[out=\angle,in=-\angle] (node03);
    
    \draw[walk] (node10) to[out=\angle,in=-\angle] (node11);
    \draw[walk] (node11) to[out=\angle,in=-\angle] (node12);
    \draw[walk] (node12) to[out=\angle,in=-\angle] (node13);
    
    \draw[walk] (node20) to[out=180-\angle,in=180+\angle] (node21);
    \draw[walk] (node21) to[out=180-\angle,in=180+\angle] (node22);
    \draw[walk] (node22) to[out=180-\angle,in=180+\angle] (node23);
    
    \draw[walk] (node01) to[out=90-\angle,in=90+\angle] (node11);
    \draw[walk] (node11) to[out=90-\angle,in=90+\angle] (node21);
    
    \draw[walk] (node03) to[out=-90+\angle,in=-90-\angle] (node13);
    \draw[walk] (node13) to[out=-90+\angle,in=-90-\angle] (node23);
    
    \draw[walk] (node20) -- (node30);
    \draw[walk] (node23) -- (node33);
    
    \draw[walk] (node02) -- (node12);
    
    \draw[public] (node20) to [out=-90-\angle,in=-90+\angle] (node00);
    \draw[public] (node00) to[out=180-\angle,in=180+\angle] (node01);
    \draw[public] (node01) to[out=180-\angle,in=180+\angle] (node03);
    \draw[public] (node03) to[out=90-\angle,in=90+\angle] (node23);
    \draw[public] (node23) to[out=-\angle,in=\angle] (node22);
    \draw[public] (node22) to[out=-\angle,in=\angle] (node20);
    
    \draw[public] (node30) -- (node33);
    
    \ifshowprices
    \draw[road,opacity=1,very thick,->,dotted] (node00) -- (node12) node[pos=0.7,left]  {2};
    \draw[road,opacity=1,very thick,->,dotted] (node02) -- (node13) node[pos=0.5,right] {4};
    \draw[road,opacity=1,very thick,->,dotted] (node00) -- (node21) node[pos=0.3,above] {1};
    \fi
    
    \ifshowpaths
    \draw[road,opacity=1,very thick,->] (node00) -- (node21); 
    \draw[walk,opacity=1,very thick,->] (node21) to[out=180-\angle,in=180+\angle] (node22); 
    
    \draw[public,opacity=1,very thick,->] (node00) -- (node23); 
    \draw[walk,opacity=1,very thick,->] (node23) to[out=180+\angle,in=180-\angle] (node22); 
    \fi
    
    \ifshowflows
    \draw[road,opacity=1,very thick,->] (node00) -- (node10); 
    \draw[road,opacity=1,very thick,->] (node10) -- (node20); 
    \draw[road,opacity=1,very thick,->] (node20) -- (node21); 
    \fi
    \ifshowrebalancing
    \draw[road,opacity=1,very thick,->] (node21) -- (node11); 
    \draw[road,opacity=1,very thick,->] (node11) -- (node01); 
    \draw[road,opacity=1,very thick,->] (node01) -- (node00); 
    \fi
\end{tikzpicture}

%% file: chapters/4_specialcase.tex
\section{The Interplay between AMoD Systems and Public Transport}\label{sec:amodgame}

We now tailor our general framework to the case where two \gls{abk:amod} operators interact with themselves and a public transport system. Further, we show how this case simplifies when a single \gls{abk:amod} operator monopolizes the market and competes only with the public transport system.  
We believe these two cases are realistic for many real-world settings, either because only one firm can operate the autonomous taxis or because (e.g., due to natural economies of scale) there are only two competing autonomous ride-hailing systems.
Our game-theoretic model offers two use cases. First, \gls{abk:amod} operators can solve the underlying game at set intervals (e.g., every hour) to determine profit-maximizing prices on short-term demand predictions. Second, the game can guide strategic decisions and policymaking in mobility systems, as we illustrate in our case studies in~\cref{sec:case study,sec:results}; e.g., municipalities can use our game-theoretic model to assess the impact of various transport pricing policies on user behavior, modal split, and costs.

\subsubsection*{General Setting}\label{subsec:general setting}
Formally, we focus on a game with three operators, all providing service to customers but with different pricing strategies. The municipality (Operator~3) operates a public transport system through the \gls{abk:pta}. Here, prices are fixed for a medium-term time horizon. Hence, we treat the municipality’s pricing strategy as fixed, i.e., $\priceStrategiesSet{3}=\{\priceStrategy{3}\}$.
With this assumption, the game consists of a simultaneous game between two \gls{abk:amod} operators (Operators~1 and 2), offering mobility services on the road network. The operators take customer requests and the \gls{abk:pta}'s prices as given and can also compute customers' optimal reaction curves for their own possible pricing strategies. Given this information, the \gls{abk:amod} operators act as strategic players maximizing their profit. 

The problem of \gls{abk:amod} operator $j$ is as follows.
The operator selects a short-term pricing strategy $\priceStrategy{j}\in\priceStrategiesSet{j}\coloneqq\nonnegativeReals^{\verticesSet{j}\times\verticesSet{j}}$ to maximize profit. All \gls{abk:amod} operators provide mobility service on the road, so $\graph{j}=\graph{\mathrm{r}}$ where $\graph{\mathrm{r}}$ is the road graph. 
Operators face a transportation system in steady state and operate a fleet of $\maxNumberVehicles{j}$ vehicles to serve the customers' transportation requests and rebalance the fleet. Hence, the \gls{abk:amod} operator's action set $\operatorsActionSpace{j}$ comprises a set of balanced flows $(\flowsSet{1},\ldots,\flowsSet{M},\flowsSet{0})$ (i.e., such that in- and out-degree of all vertices coincide) such that 
\begin{equation*}
    \begin{aligned}
        \summe{i\in\demandIndexSet,\flow{}\in\flowsSet{i}}\quad 
        \projectionRateFlow{\flow{}}\cdot\timeRoad{i,j}
        + \sum_{\flow{}\in\flowsSet{0}}
        \projectionRateFlow{\flow{}}\cdot \timeArc{j}(\projectionPathFlow{\flow{}})\leq \maxNumberVehicles{j}
    \end{aligned}
\end{equation*}
where the number of vehicles corresponding to a flow results from the multiplication of its rate and travel time. Here, $\timeRoad{i,j}$ is the time required by operator $j$ to serve the demand $\demand{i}$, assumed to be known a priori, and $\timeArc{j}:\pathsSet{\arcsSet[bar]{j}}\to\nonnegativeReals$ is a function mapping each path to its travel time.
Since travel time in urban settings mainly depends on traffic conditions and not on a single fleet's vehicles, we assume that $\timeRoad{i,j}$ and $\timeArc{j}$ are identical for all operators, and call them $\timeRoad{i}$ and $\timeArc{}$. Nonetheless, our framework can readily incorporate different travel times for different \gls{abk:amod} operators.

\subsubsection*{Customers' Reactions and Operators' Decisions}\label{subsec:reations and decisions}
\paragraph{Customer route selection}
Customers select their preferred trip through a navigation app and can choose between an \gls{abk:amod} ride~($\pathRoad{i}$) and a public transport ride combined with walking~($\pathPublicTransit{i}$):
\begin{equation*}
\begin{aligned}
    \pathRoad{i,j} &\coloneqq
    (\arc{}),\quad \arc{}\in\arcsSet[bar]{j},
    \arcSource[bar]{j}(a)=\originVertex{i}, \arcTarget[bar]{j}(a)=\destinationVertex{i},
    \\ 
    \pathPublicTransit{i} &\in 
    \shortestPath{}(\originVertex{i},\destinationVertex{i}),
\end{aligned}
\end{equation*}
where $\originVertex{i}$ and $\destinationVertex{i}$ are the origin and the destination of the $i$\textsuperscript{th} demand $\demand{i}=(\originVertex{i},\destinationVertex{i},\demandRate{i})$, respectively. 
The public transport path $\pathPublicTransit{i}$ results from the shortest path\myast{} on the union of the fully-connected public transport subgraph $\graphDefinition[bar]{3}$ and the non-controlled subgraph $\graph{0}$, computed by weighing each arc with the sum of the monetary value of time ($\valueTime$) and its fare.
Accordingly, the navigation app weighs each arc $\arc{}\in\arcsSet{0}\cup\arcsSet[bar]{3}$ with $\valueTime\cdot\timeArc{}(\arc{})$ if $\arc{}\in\arcsSet{0}$ (i.e., only walking) and $\priceStrategy{3}(\arcSource[bar]{3}(a),\arcTarget[bar]{3}(a))+\valueTime\cdot\timeArc{}(\arc{})$ if $\arc{}\in\arcsSet[bar]{3}$ (i.e., public transport and walking).
The \gls{abk:amod} path $\pathRoad{i,j}$ results from the arc in the fully-connected \gls{abk:amod} operator subgraph $\graphDefinition[bar]{j}$ having the origin and the destination of the demand as the source and sink vertex, respectively.
This definition holds without loss of generality as each vertex in the non-controlled subgraph, on which demands are placed, can be associated to a vertex in an \gls{abk:amod} operator subgraph.
Accordingly, the action space of the customers reads
$\customersActionSpace{i} = \{\left.
    \reactionCurve{}\in\nonnegativeReals^{\pathsSetRequest{\demand{i}}}
    \,\right|\,
    \reactionCurve{}(\path{})=0\,\forall\,\path{}\in\pathsSetRequest{\demand{i}}\setminus\{\pathRoad{i,1},\pathRoad{i,2},\pathPublicTransit{i}\}, 
    \sum_{\path{}\in\pathsSetRequest{\demand{i}}}\reactionCurve{i}(\path{})=\demandRate{i}
\}.$
We consider rational customers who minimize their total cost, given by the sum of fares paid and their monetary value of time ($\valueTime$), so that the cost associated with a reaction curve $\reactionCurve{}\in\customersActionSpace{i}$, given the pricing strategies of the \gls{abk:amod} operators and of the municipality, reads
\begin{equation*}
    \begin{aligned}[t]
    J_i(\reactionCurve,\priceStrategy{1},\priceStrategy{2})
    =
    &\sum_{j=1}^2
    (\priceStrategy{j}(\originVertex{i},\destinationVertex{i}) + \valueTime \cdot \timeRoad{i})\cdot \reactionCurve{}(\pathRoad{i,j}) \\
    &+(\pricePublicTransit{i} + \valueTime \cdot \timePublicTransit{i})\cdot \reactionCurve{}(\pathPublicTransit{i}),
    \end{aligned}
\end{equation*}
with $\timeRoad{i}$ and $\timePublicTransit{i}$ being the travel times for the demand $\demand{i}$ when choosing either an \gls{abk:amod} or a public transport ride, and $\pricePublicTransit{i}\coloneqq\sum_{a\in\pathPublicTransit{i}\cap\arcsSet[bar]{3}}\priceStrategy{3}(\arcSource[bar]{3}(a),\arcTarget[bar]{3}(a))$ being the price related to the path $\pathPublicTransit{i}$. We assume $\timeRoad{i}$ and $\timePublicTransit{i}$ to be distinct, but allow them to be arbitrarily close. Then, with \cref{definition:customer equilibrium}, the reaction curve $\reactionCurve{i}$ of a homogeneous demand is optimal if it minimizes $J_i(\reactionCurve,\priceStrategy{1},\priceStrategy{2})$.
In reality, individual customers have different values of time such that demands become heterogeneous. We consider such heterogeneity by defining $\valueTime\sim\mathbb{P}$ to be dependent on some probability distribution~$\mathbb{P}$. Moreover, to account for variations in \gls{abk:amod} operators' overall prices, we introduce a zero-mean uncertainty. This uncertainty accounts for customers valuing the service of a given operator more than the other and for stochasticity in an \gls{abk:amod} operator service. Accordingly,
\begin{equation*}
    \label{eq:heterogeneous reaction curve}
    \begin{aligned}[t]
    \reactionCurve[]{i}
    \!=\!
    \mathbb{E}\Bigg[
    	\argmin{\reactionCurve{}\in\customersActionSpace{i}}\!\!
    	&\sum_{j=1}^{2}
    	\!\left(\priceStrategy{j}(\originVertex{i},\destinationVertex{i}) + \valueTime\! \cdot \!\timeRoad{i} + \varepsilon_j \right)\cdot \reactionCurve{}(\pathRoad{i,j})
    	\\
    	&+
    	\left(\pricePublicTransit{i} + \valueTime \!\cdot \!\timePublicTransit{i}\right)\cdot \reactionCurve{}(\pathPublicTransit{i})
    	\Bigg],
    \end{aligned}	
\end{equation*}
where the expected value is with respect to $\valueTime$ and $\varepsilon_j$.
In this work, we proceed with uniformly distributed values of time. We also assume that $\varepsilon_1=0$ and that $\varepsilon_2$ is uniformly distributed and zero-mean, and choose the parameter of the distribution such that it results in realizations that are one order of magnitude smaller than realizations of the value of time. This second probability distribution allows us to distinguish the two \gls{abk:amod} operators; else, the reaction $\reactionCurve{i}$ would be ill-defined. The resulting reaction curves can also be seen as a special case of general discrete choice models (e.g., see \cite{Train2009,Cantarella1997}).

\paragraph{\gls{abk:amod} operator profit maximization}
The \gls{abk:amod} operators aim to maximize their profit. Given rational reactions of the customers $\reactionCurve{i}\in\customersEquilibria{i}(\priceStrategy{1},\priceStrategy{2},\priceStrategy{3})$ the profit of \gls{abk:amod} operator $j$ equals the difference of revenue (Eq. \eqref{eq:operator revenue}) and costs (Eq. \eqref{eq:operator cost}).

\subsubsection*{Problem Analysis}\label{subsec:analysis}
In this section, we prove existence of an equilibrium and give a tractability result. We start by providing an expression for the reaction curves.
\begin{lemma}[Rational reaction]
\label{lemma:reaction curve}
Rational reaction curves are characterized by the pointwise minimum of affine functions:
\begin{equation}\label{eq:reaction curve}
\reactionCurve{i}(\pathRoad{i,j})
=
\max\{\min\{h_{i,j}(\priceStrategy{j}(\originVertex{i},\destinationVertex{i})), \demandRate{i}\}, 0\},
\end{equation}
where 
\begin{equation*}
\begin{aligned}
h_{i,j}(\priceStrategy{j}(\originVertex{i},\destinationVertex{i}))	
&\coloneqq\min\{
\begin{aligned}[t]
&m_{i,1}\priceStrategy{j}(\originVertex{i},\destinationVertex{i})+q_{i,1}, \\
&m_{i,2}\priceStrategy{j}(\originVertex{i},\destinationVertex{i})+q_{i,2}(\priceStrategy{-j}(\originVertex{i},\destinationVertex{i})), \\
&m_{i,3}\priceStrategy{j}(\originVertex{i},\destinationVertex{i})+q_{i,3}(\priceStrategy{-j}(\originVertex{i},\destinationVertex{i}))\},
\end{aligned}
\end{aligned}
\end{equation*}
and $m_{i,1},m_{i,2},m_{i,3}<0$. 
Define $p_{i,j}^\mathrm{min}$ by $h_{i,j}(p_{i,j}^\mathrm{min})=\demandRate{i}$ and $p_{i,j}^\mathrm{max}$ by $h_{i,j}(p_{i,j}^\mathrm{max})=0$, both well-defined and uniformly bounded.
Then, $h_{i,j}$ defines a homeomorphism between $\priceStrategy{j}(\originVertex{i},\destinationVertex{i})\in[p_{i,j}^\mathrm{min},p_{i,j}^\mathrm{max}]$ and $\phi(\pathRoad{i,j})\in[0,\demandRate{i}]$. Moreover, $q_{i,2}$ and $q_{i,3}$ change linearly with $\priceStrategy{-j}(\originVertex{i},\destinationVertex{i})$, such that $h_{i,j}$ changes continuously with the price of the adversary \gls{abk:amod}.
\end{lemma}

\cref{lemma:reaction curve} gives three major insights.
First, the existence of a homeomorphism between prices and reaction curves allows us to use reaction curves as optimization variables rather than the pricing strategy itself.
Second, it shows that reaction curves are concave and therefore amenable to convex reformulations. 
Third, it tells us that reaction curves are continuous in the pricing strategy of the adversary \gls{abk:amod} operator, which will be one of the key ingredients of the existence result at the end of this section.
As a result, we obtain a tractable reformulation for the computation of best responses. 

\begin{proposition}[Best response]
\label{proposition:best response}
	The best response of each \gls{abk:amod} operator results from a tractable convex second-order conic program. Specifically, we obtain best response prices for origin-destination pairs subject to demand from $h_{i,j}^{-1}(x_i)$, whereby $x_i$ results from
\begin{subequations}\label{eq:best response socp}
\begin{align}
\max_{\substack{x_{i}\in [0,\demandRate{i}], \\ f_0\in\mathbb{R}^{\cardinality{\arcsSet{j}}}}}
&\sum_{i=1}^{\demandNumber} r_i - \sum_{a\in\arcsSet{j}}\operationCostFunction{j}(\arc{})\left[\sum_{i=1}^{\demandNumber} f_{i,j}^\ast x_i + f_{0}\right]_a\label{eq:objective}\\
\text{subject to }
&0\leq r_i \leq x_{i}\cdot(x_i-q_{i,k})/m_{i,k}\:\forall k\in\{1,2,3\}\label{eq:constraints revenue}\\
&B^\top\left(\sum_{i=1}^{\demandNumber} f^\ast_{i,j} x_{i}+f_0\right)=0 \label{eq:constraints veh conservation} \\
&\sum_{i=1}^{\demandNumber}\timeRoad{i,j}x_{i} + \sum_{a\in\arcsSet{j}}\timeFunction{j}(a)[f_0]_a\leq \maxNumberVehicles{j},\label{eq:constraints fleet}
\end{align}
\end{subequations}
where $f^\ast_{i,j}$ corresponds to the shortest path from the demand's origin to the destination, $[f]_a$ denotes the entry corresponding to arc $a\in\arcsSet{j}$, and $B$ is the incidence matrix of the \gls{abk:amod} operator's graph. 
\end{proposition}
In words, \eqref{eq:objective} is the profit, consisting of the revenue, captured by the auxiliary variable $r_i$, and the costs for serving the customers and for rebalancing the fleet. Constraint~\eqref{eq:constraints revenue} ensures consistency between the revenue and $x_i$, via the demand curve~\eqref{eq:reaction curve}. Constraint~\eqref{eq:constraints veh conservation} ensures that flows are balanced; i.e., the number of vehicles entering and exiting a node coincide. Finally, constraint~\eqref{eq:constraints fleet} takes care of the fleet size. 
The result of \cref{lemma:reaction curve} is crucial to obtain a tractable best response: straightforwardly using prices as optimization variables would lead to a non-convex optimization problem. Specifically, the vehicle conservation constraint \eqref{eq:constraints veh conservation} would in this case be non-convex, as $x_i$ is not an affine function of the pricing strategy.
With \cref{lemma:reaction curve,proposition:best response}, we can establish existence of an equilibrium.

\begin{theorem}[Existence of an equilibrium]\label{theorem:existence}
The game admits an equilibrium.
\end{theorem}

As it is typically the case in game theory, the equilibrium is not unique. In practice, one may iteratively run best response to compute one of the equilibria. It is well-known that best response might not converge, yet it converges to an equilibrium if it converges.

\subsubsection*{Single Operator Case}\label{subsec:single operator}
In the degenerate case of a single \gls{abk:amod} operator, the game reduces to a quadratic optimization problem. We summarize the result in the following theorem.

\begin{theorem}[Equilibria in the single \gls{abk:amod} operator case]
	\label{theorem:equilibrium linear}
	The game admits an equilibrium. 
	Further, consider $(\priceStrategy[star]{1},\priceStrategy[star]{3})\in\priceStrategiesSet{1}\times\priceStrategiesSet{3}$ with 
	\begin{enumerate}[label=(\textit{\roman*})]
		\item $\priceStrategy[star]{1}(o,d)=0$ if there is no demand from $o$ to $d$, i.e., $(o,d,\demandRate{})\notin\demandsSet{}$ for all $\demandRate{}\in\nonnegativeReals$,
		\item $\priceStrategy[star]{1}(o,d)=[\rho^\star]_i$ if there is a demand from $o$ to $d$, i.e., $\demand{i}=(o,d,\demandRate{})$ for some $\demandRate{}\in\nonnegativeReals$, whereby $\rho^\star\in\nonnegativeReals^{\demandNumber}$ results from solving a convex quadratic program, 
		\item $\priceStrategy[star]{3}=\priceStrategy{3}$.
	\end{enumerate}
    Then, $(\priceStrategy[star]{1},\priceStrategy[star]{3})$ is an equilibrium. Moreover, all equilibria are equivalent in the \gls{abk:amod} operator's and the municipality's profit as well as in the demands' optimal reaction curves.
\end{theorem}

%% file: chapters/5_casestudy.tex
\section{Case Study: Berlin, Germany}\label{sec:case study}
Our case study bases on a real-world setting for the city center of Berlin, Germany.
We derive the road network from OpenStreetMap data \cite{HaklayWeber2008} and infer the public transport network, consisting of U-Bahn, S-Bahn, tram, and bus lines, together with its schedules from GTFS data \cite{BerlinOpenData}.
We use demand data from an existing case study \cite{HorniAxhausen2016,Ziemke2017}. Since the authors considered a representative ten percent sample of the population, we scale the rate of each demand by a factor of 10 so that our dataset consists of \SI{129560}{} travel requests and a total demand rate of \SI{18} customers per second for a two-hour time horizon. The average length of a trip is \SI{4.9}{\kilo\meter}.

We compute the \gls{abk:amod} service travel time $\timeRoad{i}$ for demand $\demand{i}$ as the sum of the net road time\footnote{As we use a distance-based cost, the shortest path $f_{i,j}^\ast$ is identical for all \gls{abk:amod} operators and therefore $f_i^\ast$ is well-defined.} $\sum_{\arc{}\in\projectionPathFlow{\flow[star]{i}}}\timeArc{}(\arc{})$
and an average waiting time of $\SI{3}{\minute}$, according to today's waiting time for ride-hailing companies \cite{Mosendz2014}.
To account for congestion effects, we increase each arc's nominal travel time, given by the length over the free-flow velocity, by 56\%, corresponding to the evening peak congestion level of a workday in Berlin \cite{TomTom}.
We consider an average walking velocity of \SI{1.4}{\meter\per\second}. We compute public transport travel times based on the public transport schedules and consider a waiting time at a station of half the average time interval between two trips (\SI{5}{\minute} for U-Bahn and S-Bahn, \SI{7}{\minute} for trams, and \SI{10}{\minute} for buses) and \SI{60}{\second} walking-to-station and station-to-walking time for the U-Bahn and S-Bahn.

To calculate costs for \gls{abk:amod} operators, we consider autonomous electrified taxis with distance-based operation cost of $\operationCostDistance{j}=\SI{0.34}{\usd\per\kilo\meter}$ \cite{Boesch2017}.
This cost accounts for both variable costs, in terms of electrical energy, depreciation, maintenance, and fixed costs, in terms of acquisition and insurance. 
We impose a maximum of \SI{8373}{} licenses for autonomous vehicles, reflecting the \SI{8373}{} taxi concessions released by the Berlin municipality in 2018 \cite{Neumann2019}. We consider three cases: \emph{(i)} licenses are equally split between two identical \gls{abk:amod} operators, (\emph{ii}) licenses are split between two non-identical \gls{abk:amod} operators such that one fleet is \SI{50}{\percent} larger than the other, and (\emph{iii}) all licenses are assigned to a single \gls{abk:amod} operator. 
Consistently with the fares in Berlin, 
we set the price of a public transport ride to \SI{2.80}{\eur}, corresponding to \SI{3.12}{\usd}.
In line with \cite{Endorf2016} and \cite{Wadud2017}, we assume the customers' value of time to be uniformly distributed between \SI{10}{\usd\per\hour} and \SI{17}{\usd\per\hour}.

Based on this case study, we study different settings to investigate the interplay between \gls{abk:amod} systems and the municipality.
Besides our basic setting (S1), we analyze potential \gls{abk:amod} operator strategies (S2--S3) to influence the equilibrium of the system. 
Settings (S4--S5) analyze how a municipality may steer the system equilibrium using regulations that influence the \gls{abk:amod} operator decisions.
Specifically:

\textit{Basic setting:} We analyze the basic setting of our case study for both two identical \gls{abk:amod} operators, two non-identical \gls{abk:amod} operators, and a single \gls{abk:amod} operator, and compare them to the setting without \gls{abk:amod} operators, whereby customers can only walk or use public transport.

\textit{Fleet size:} We investigate the impact of the \gls{abk:amod} fleet size. To this end, we perform a parametric study, varying the allowed fleet size from \SI{1000}{} to \SI{23000}{} in intervals of \SI{2000}{} vehicles. For identical \gls{abk:amod} operators we split the licenses equally, for non-identical \gls{abk:amod} operators we split them such that one fleet is \SI{50}{\percent} larger than the other. 

\textit{Public transport price:} We perform a parametric study to quantify the impact of the public transport price. Specifically, we vary the fares of public transportation between \SI{0}{\usd} and \SI{6}{\usd} per ride, with a step width of \SI{0.5}{\usd}.

\textit{\gls{abk:amod} service tax:} In line with recent discussions on ride-hailing service taxes, we analyze the impact of an additional percentage tax on the revenue of each trip served by an \gls{abk:amod} operator \cite{Welle2018}. 
We consider taxes ranging from \SI{0}{\percent} to \SI{100}{\percent}, with a step width of \SI{10}{\percent}.

For all settings with two \gls{abk:amod} operators, we find a game equilibrium by iteratively computing best responses.
For our base case equilibrium, the computational time on commodity hardware amounts to 8 minutes (with the computation of an individual best response taking roughly 1 second).
For the single \gls{abk:amod} operator settings, we find an equilibrium via the quadratic optimization problem resulting from \cref{theorem:equilibrium linear}.

%% file: chapters/6_results.tex
\FloatBarrier
\section{Results}\label{sec:results}
\paragraph{Base case equilibrium}
\cref{fig:e1} details the modal share and the profit per trip distribution for the \gls{abk:amod} operator(s) for cases of two identical \gls{abk:amod} operators (\cref{fig:e1 competitive identical}), two non-identical \gls{abk:amod} operators (\cref{fig:e1 competitive nonidentical}), and a single \gls{abk:amod} operator (\cref{fig:e1 single}). In all cases, the usage of different transportation modes splits nearly equally between \gls{abk:amod} and public transport, while only a small share of customers opts to complete a trip solely by walking. While \cref{fig:e1} suggests that trips split rather equally between public transport and \gls{abk:amod} services from a macroscopic perspective, the opposite is the case when analyzing the solution from a microscopic perspective. Often, either the \gls{abk:amod} service or the public transport provides a cost-optimal solution, independently of the customer's value of time, so that several origin-destination pairs are served completely by public transport, one of the \gls{abk:amod} operator, or both \gls{abk:amod} operators (cf. \cref{tab:local cannibalization}). For instance, in the single \gls{abk:amod} case, \SI{20.9}{\percent} of all trips are served solely by \gls{abk:amod}, \SI{38.3}{\percent} of all trips are served solely by public transport, and \SI{4.3}{\percent} of all trips are completed solely by walking.

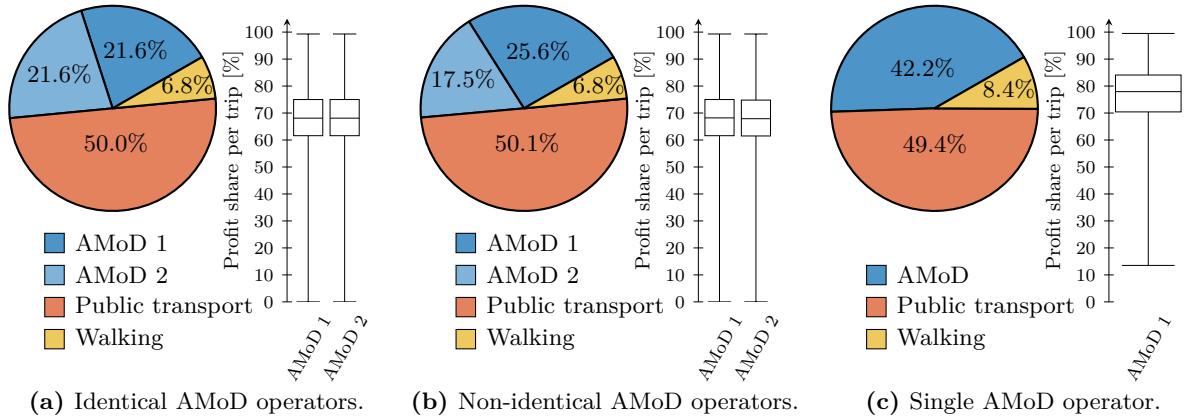
\begin{figure}[!b]
    \centering
    \captionsetup[subfigure]{skip=-2pt}
    \def\radiusPlot{1.6}
	\begin{subfigure}[t]{0.325\columnwidth}
		\centering 	
		\def\amodfirstShare{21.6}
		\def\amodsecondShare{21.6}
		\def\ptShare{50.0}
		\def\walkingShare{6.8}
		\def\lowerQuantileFirst{0}
		\def\midlowerQuantileFirst{61.6}
		\def\midQuantileFirst{68.1}
		\def\midhigherQuantileFirst{75.0}
		\def\higherQuantileFirst{99.3}
		\def\lowerQuantileSecond{0}
		\def\midlowerQuantileSecond{61.6}
		\def\midQuantileSecond{68.1}
		\def\midhigherQuantileSecond{75.0}
		\def\higherQuantileSecond{99.3}
		\input{graphics/tikz/e1pieShareTwo.tex}
		\caption{Identical \gls{abk:amod}.}
		\label{fig:e1 competitive identical}		
	\end{subfigure}
	\hfill 
	\begin{subfigure}[t]{0.325\columnwidth}
		\centering
		\def\amodfirstShare{25.6}
		\def\amodsecondShare{17.5}
		\def\ptShare{50.1}
		\def\walkingShare{6.8}
		\def\lowerQuantileFirst{0}
		\def\midlowerQuantileFirst{61.6}
		\def\midQuantileFirst{68.2}
		\def\midhigherQuantileFirst{75.0}
		\def\higherQuantileFirst{99.3}
		\def\lowerQuantileSecond{0}
		\def\midlowerQuantileSecond{61.5}
		\def\midQuantileSecond{67.9}
		\def\midhigherQuantileSecond{74.8}
		\def\higherQuantileSecond{99.3}
		\input{graphics/tikz/e1pieShareTwo.tex}
		\caption{Non-ident. \gls{abk:amod}.}
		\label{fig:e1 competitive nonidentical}		
	\end{subfigure}
	\hfill 
	\begin{subfigure}[t]{0.325\columnwidth}
		\centering 
		\def\amodShare{42.2}
		\def\ptShare{49.4}
		\def\walkingShare{8.4}
		\def\lowerQuantile{13.5}
		\def\midlowerQuantile{70.4}
		\def\midQuantile{77.9}
		\def\midhigherQuantile{84.1}
		\def\higherQuantile{99.5}
		\input{graphics/tikz/e1pieShareSingle.tex}
		\caption{Single \gls{abk:amod}.}
		\label{fig:e1 single}
	\end{subfigure}
	\caption{Modal share (left) and statistics of the profit share per trip, i.e., the percentage share of the revenue of each trip which remains as profit for the \gls{abk:amod} operator(s) (right).}
	\label{fig:e1}
\end{figure}

\begin{table}[t]
    \scriptsize 
    \centering 
    \setlength{\tabcolsep}{3pt}
    \begin{tabular}{rcccccc}
        \toprule 
         &
         \multicolumn{6}{c}{Percentage of trips served by}
         \\ 
         \cmidrule(lr{1em}){2-7}
         &
         \makecell{PT \\ only}
         & 
         \makecell{\gls{abk:amod} \\ only}
         & 
         \makecell{\gls{abk:amod} 1 \\ only}
         &
         \makecell{\gls{abk:amod} 2 \\ only}
         &
         \makecell{walking \\ only}
         &
         mixed \\ \midrule 
         Identical \gls{abk:amod} &
         \SI{44.8}{\percent} & \SI{23.6}{\percent} & \SI{0.0}{\percent} & \SI{0.0}{\percent} & \SI{5.2}{\percent} & \SI{26.4}{\percent}\\
         Non-ident. \gls{abk:amod}  &
         \SI{44.7}{\percent} & \SI{29.0}{\percent} & \SI{0.5}{\percent} & \SI{0.0}{\percent} & \SI{5.2}{\percent} & \SI{21.1}{\percent} \\
         Single \gls{abk:amod} &
         \SI{38.3}{\percent} & \SI{20.9}{\percent} & \SI{20.9}{\percent} & -- & \SI{4.3}{\percent} & \SI{36.5}{\percent} \\
         \bottomrule 
    \end{tabular}
    \caption{Local cannibalization of \gls{abk:amod} operators. Results confirm that the microscopic trend  deviates from the macroscopic one: many trips do not show a modal share and are entirely served by only one mode of transportation.}
    \label{tab:local cannibalization}
\end{table}

The overall profit share differs only insignificantly from the numbers reported for the identical operator case.
In the single operator case, the profit share per trip distribution is even higher, showing a \SI{10}{\percent} increase for its median and a more than \SI{10}{\percent} increase for its minimum compared to the multiple operator case. In this case, the overall profit share increases accordingly and amounts to \SI{73.9}{\percent} of the revenue, while \SI{26.1}{\percent} are used to cover the costs.
In the case of two identical \gls{abk:amod} operators, we further observe that the total \gls{abk:amod} share is equally split between the two operators and the profit share per trip distribution is equal, too. Each \gls{abk:amod} operator's overall profit share is \SI{64.7}{\percent}, while \SI{35.3}{\percent} of the revenue is used to cover the costs.
The profit share per trip distribution remains equal between both \gls{abk:amod} operators in the non-identical case, but the larger operator obtains a \SI{45}{\percent} larger modal share and a \SI{48}{\percent} larger total profit. These effects correlate with the operator's \SI{50}{\percent} larger fleet size. 

These results allow for the following interpretation.
In all cases, \gls{abk:amod} operators are able to realize a high profit share as they can exploit the limits of the static public transport pricing strategy effectively.
In the degenerate single operator case, the \gls{abk:amod} operator faces no additional competition and therefore achieves higher profit shares than in the multiple operator cases. Here, the limited fleet size constrains her actions and the modal share at equilibrium results accordingly.
In the multiple operator cases, \gls{abk:amod} operators face competition among each other, which leads to lower profit share realizations. However, the fleet size remains the constraining quantity. Accordingly, the \gls{abk:amod} operators still obtain similar relative profit shares and their share of the modal split differs proportionally to the difference in their fleet sizes. 

Compared to a setting without AMoD operators (i.e., with only public transport or walking), the case of identical AMoD operators yields \SI{7.7}{\percent} lower customers' transportation costs. Not all citizens, though, profit in the same way. As shown in~\cref{tab:local cannibalization}, some trips do not benefit from the service of the AMoD operators and are still completed walking or using public transportation; in contrast, all other trips are \SI{17.5}{\percent} cheaper and trips only served by AMoD are \SI{25}{\percent} cheaper. 
In the monopolistic case, transportation costs only decrease by \SI{2.2}{\percent} for all trips and by \SI{7.7}{\percent} for trips served only by AMoD, which correlates with the higher profit share discussed before. 

\paragraph{\gls{abk:amod} related impact factors}
\cref{fig:e2} shows the modal split and the \gls{abk:amod} operators' profit, depending on the \gls{abk:amod} fleet size, for the identical (\cref{fig:e2 competitive identical}), non-identical (\cref{fig:e2 competitive nonidentical}), and the single (\cref{fig:e2 single}) \gls{abk:amod} operator case.
As can be seen, the \gls{abk:amod} fleet size heavily influences the modal split, resulting in \gls{abk:amod} shares that vary between \SI{7.3}{\percent} and \SI{80.2}{\percent}. 

\begin{figure}[!t]
	\centering
	\captionsetup[subfigure]{skip=-10pt}
	\def\barwidthPlot{4pt}
	\def\xlabelPlot{Fleet size [$\times 10^3$ vehicles]}
	\def\xticksPlot{0,4,...,23}
    \newboolean{showylabelleft}
	\newboolean{showylabelright}
	\newboolean{showlegend}
	\begin{subfigure}[t]{0.325\columnwidth}
	    \centering
	    \setboolean{showylabelleft}{true}
		\setboolean{showylabelright}{true}
		\setboolean{showlegend}{true}
		\pgfplotstableread[col sep=comma]{graphics/matlab/BerlinFleetSize/BerlinFleetSizeCompetitive.csv}{\dataTable}
		\input{graphics/tikz/eScanTwo.tex}
		\caption{Identical \gls{abk:amod}.}
		\label{fig:e2 competitive identical}
	\end{subfigure}
	\hfill 
	\begin{subfigure}[t]{0.325\columnwidth}
	    \centering
	    \setboolean{showylabelleft}{true}
		\setboolean{showylabelright}{true}
		\setboolean{showlegend}{false}
		\pgfplotstableread[col sep=comma]{graphics/matlab/BerlinFleetSize/BerlinFleetSizeCompetitive_Nonidentical.csv}{\dataTable}
		\input{graphics/tikz/eScanTwo.tex}
		\caption{Non-ident. \gls{abk:amod}.}
		\label{fig:e2 competitive nonidentical}
	\end{subfigure}
	\hfill 
	\begin{subfigure}[t]{0.325\columnwidth}
	    \centering
	    \setboolean{showylabelleft}{true}
		\setboolean{showylabelright}{true}
		\setboolean{showlegend}{false}
		\pgfplotstableread[col sep=comma]{graphics/matlab/BerlinFleetSize/BerlinFleetSize.csv}{\dataTable}
		\input{graphics/tikz/eScanSingle.tex}
		\caption{Single \gls{abk:amod}.}
		\label{fig:e2 single}
	\end{subfigure}
	\caption{Impact of the fleet size on the modal share and the \gls{abk:amod} operators' profit.}
	\label{fig:e2}
\end{figure}
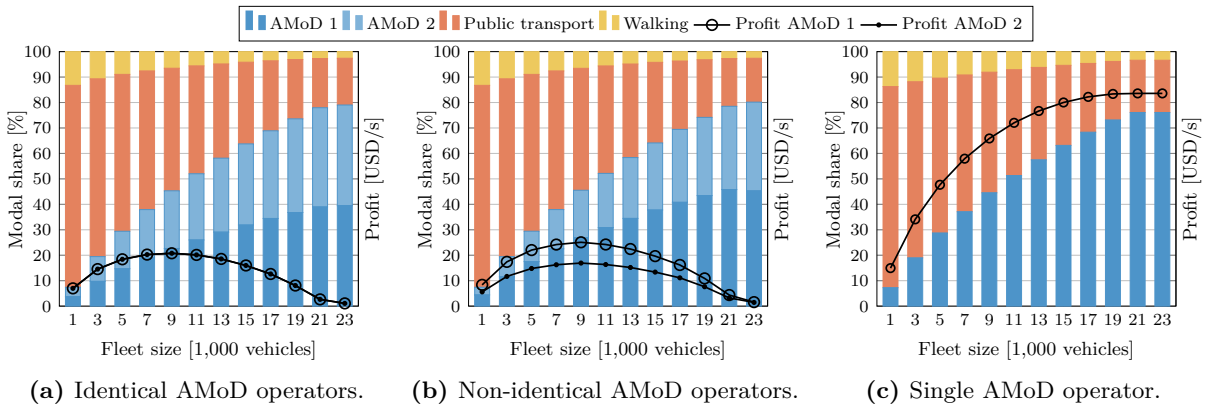 

While the overall modal share of the \gls{abk:amod} operators remains similar in all three cases, we observe a major difference in the \gls{abk:amod} operators' profit between the two and the single \gls{abk:amod} operator case. In the competitive cases (\cref{fig:e2 competitive identical,fig:e2 competitive nonidentical}), a larger fleet size eventually leads to decreasing profits. Here, the profit of both operators remains equal in the identical operators case. In the non-identical operator case, the operator with the larger fleet size gains a larger profit.
In the monopolistic case (\cref{fig:e2 single}), the profit of the \gls{abk:amod} operator monotonically increases up to a fleet size of \SI{21000}{} vehicles. Then, no more additional requests are served by the \gls{abk:amod} system as the last \SI{23.9}{\percent} remain unprofitable due to high operational and rebalancing costs. Notably, the \gls{abk:amod} operator's profit begins to nearly stagnate already for fleet sizes above \SI{15000}{} vehicles. Overall,  the \gls{abk:amod} profit shows a decreasing utility margin as the operator serves the most profitable rides first, before accepting less profitable rides. 

These results allow for the following interpretation. In the monopolistic case, the \gls{abk:amod} operator does not face competition. Accordingly, an increasing fleet size enlarges her action space and allows to directly select profit-maximizing pricing strategies. In the competitive cases, a larger fleet size still enlarges each \gls{abk:amod} operators action space, but also increases the competition among the operators. Thus, both operators are forced to lower their prices in order to attract customers, which decreases their profits accordingly. Interestingly, the larger \gls{abk:amod} operator faces a steeper profit decrease compared to the smaller \gls{abk:amod} operator in the non-identical case.
Moreover, competition naturally limits fleet size. We observe that rational \gls{abk:amod} operators would operate \SI{9000}{} vehicles (\SI{4500}{} vehicles each for identical \gls{abk:amod} operators, \SI{5400}{} and \SI{3600}{} vehicles for non-identical \gls{abk:amod} operators) in competitive cases. This steady state results arise from decreasing profits that operators face in a competitive setting. Remarkably, it reflects the current real-world situation of Berlin's taxi market in which \SI{8373}{} licensed taxis offer mobility service. In the monopolistic case, a rational \gls{abk:amod} operator would operate a fleet size of more than \SI{20000}{} vehicles.

\paragraph{Public transport impact factors}
\cref{{fig:e5}} shows the development of the modal split based on the public transport price.
In the two operators cases (\cref{fig:e5 competitive identical,fig:e5 competitive nonidentical}), the public transport price only slightly affects the modal share: free or cheaper public transport does not significantly reduce the \gls{abk:amod} modal share and higher public transport prices do not shift customers from public transport to the \gls{abk:amod} system. Instead, the modal share of walking increases for higher public transport prices while the \gls{abk:amod} modal share remains equal.
\begin{figure}[!t]
	\centering
	\captionsetup[subfigure]{skip=-10pt}
	\def\barwidthPlot{4pt}
	\def\xlabelPlot{Public transport price [\si{\usd}]}
	\def\xticksPlot{0,1,...,6}
	\newboolean{showylabelleft}
	\newboolean{showylabelright}
	\newboolean{showlegend}
	\begin{subfigure}[t]{0.325\columnwidth}
		\centering
		\setboolean{showylabelleft}{true}
		\setboolean{showylabelright}{true}
		\setboolean{showlegend}{true}
		\pgfplotstableread[col sep=comma]{graphics/matlab/BerlinPTPrice/BerlinPTPriceCompetitive.csv}{\dataTable}
		\input{graphics/tikz/eScanTwo.tex}
		\caption{Identical \gls{abk:amod}.}
		\label{fig:e5 competitive identical}
	\end{subfigure}
	\hfill  
	\begin{subfigure}[t]{0.325\columnwidth}
		\centering
		\setboolean{showylabelleft}{true}
		\setboolean{showylabelright}{true}
		\setboolean{showlegend}{false}
		\pgfplotstableread[col sep=comma]{graphics/matlab/BerlinPTPrice/BerlinPTPriceCompetitive_Nonidentical.csv}{\dataTable}
		\input{graphics/tikz/eScanTwo.tex}
		\caption{Non-ident. \gls{abk:amod}.}
		\label{fig:e5 competitive nonidentical}
	\end{subfigure}
	\hfill
	\begin{subfigure}[t]{0.325\columnwidth}
		\centering
		\setboolean{showylabelleft}{true}
		\setboolean{showylabelright}{true}
		\setboolean{showlegend}{false}
		\pgfplotstableread[col sep=comma]{graphics/matlab/BerlinPTPrice/BerlinPTPrice.csv}{\dataTable}
		\input{graphics/tikz/eScanSingle.tex}
		\caption{Single \gls{abk:amod}.}
		\label{fig:e5 single}
	\end{subfigure}
	\caption{Impact of the public transport price on the modal share and the \gls{abk:amod} operators' profit.}
	\label{fig:e5}
\end{figure}
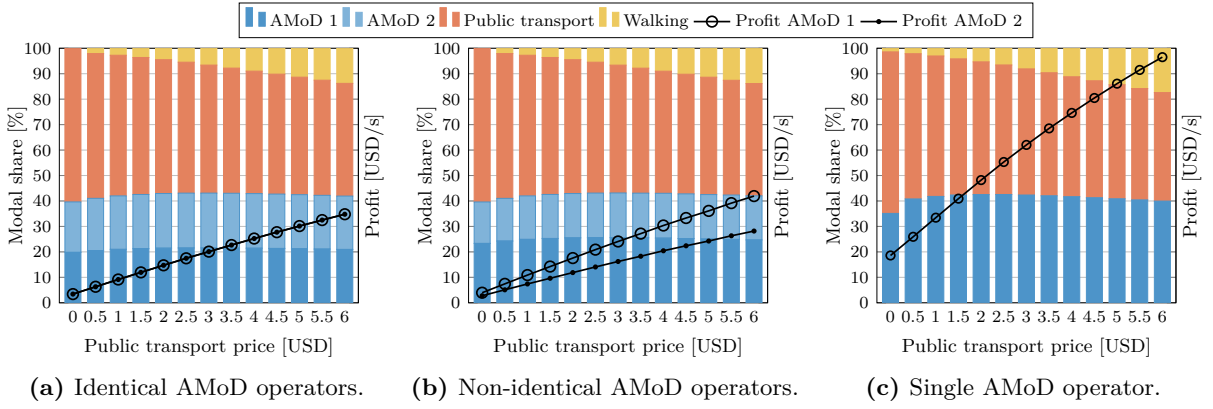
Revenue and profit of the \gls{abk:amod} operators constantly increase with higher public transport prices. These effects occur because a higher public transport price allows the \gls{abk:amod} operators to raise prices, without affecting customer decisions or their operations.
In the case of a single \gls{abk:amod} operator (\cref{fig:e5 single}), free public transport reduces the share of \gls{abk:amod} services in the system to \SI{35}{\percent}. For a large range of prices (\SI{1}{\usd} to \SI{3.5}{\usd}), the \gls{abk:amod} share remains constant, whereas it even decreases for public transport prices above \SI{3.5}{\usd} (which correlates with an increased average distance between the origin and the destination of the customers served). This slight decrease results from the \gls{abk:amod} operator focusing on less but more profitable trips.
Revenue and profit of the \gls{abk:amod} operator constantly increase with higher public transport prices analogously to the multiple operator cases. Yet, the profit of the \gls{abk:amod} operator grows roughly \SI{25}{\percent} faster than the sum of the profits of the two identical \gls{abk:amod} operators. 

These results allow for the following interpretation. In the single \gls{abk:amod} operator case, no competition exists, such that the \gls{abk:amod} operator selects the pricing strategy that serves only her profit-maximizing trips. Here, the set of profit-maximizing trips does not necessarily correlate with the highest possible modal share, which explains the decrease of the \gls{abk:amod} modal share for low and high public transport prices. In the multiple \gls{abk:amod} operator cases, the two \gls{abk:amod} operators impact each other, making the selection of the monopolistic pricing strategy profit-suboptimal. Accordingly, each operator is forced to serve more customers to increase her profit, which stabilizes the \gls{abk:amod} modal share independent of the public transport price and, at the same time, constraints both operators prices and profits. 
One comment with these findings is in order. We currently do not model a customer's option to use no transportation mode at all, i.e., to stay home. In practice, customers would decline both services if prices increased too drastically, which would limit the revenue and profit increase for both the \gls{abk:amod} operator and the public transport. This behavior is partly reflected in an increasing shift to pure walking trips. However, here the model implies that a customer's individual price threshold to resign from a trip relates to her general value of time. Although the tendencies shown are not accurate for very high public transport fees, the dynamics shown for reasonable changes in public transport prices remain valid.

\cref{fig:e7} shows the impact of a service tax, imposed on the \gls{abk:amod} operators' revenue, i.e., on the revenue of each individual trip. As can be seen, a service tax up to \SI{60}{\percent} does not yield a change in the modal share but decreases the \gls{abk:amod} operators' profit. This effect results in both the two operator cases and the single operator case from the high profit share that the \gls{abk:amod} operator earns in the basic scenario. With a service tax above \SI{60}{\percent}, the modal share changes because the \gls{abk:amod} operator begins to refuse services that become unprofitable.    
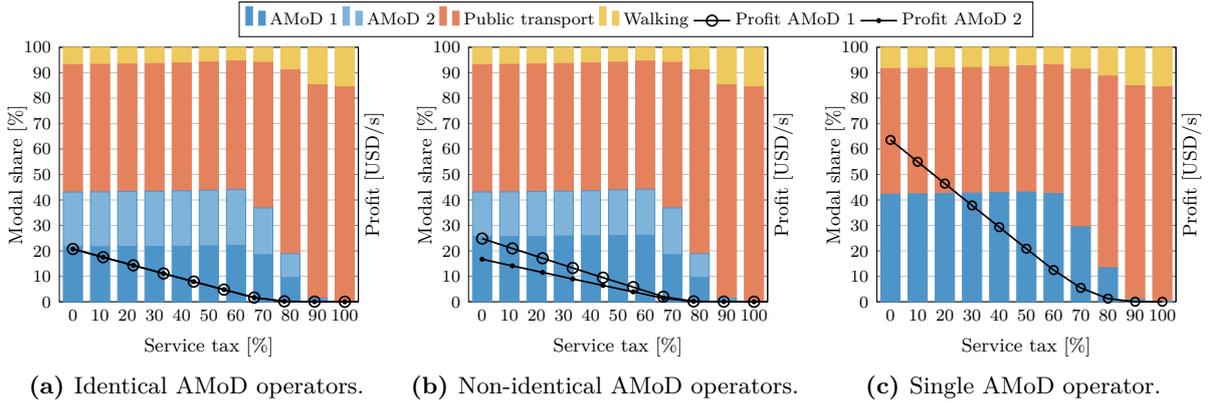
\begin{figure}[!t]
    \captionsetup[subfigure]{skip=-10pt}
	\centering 
	\def\barwidthPlot{5pt}
	\def\xlabelPlot{Service tax [\%]}
	\def\xticksPlot{0,20,...,100}
	\newboolean{showylabelleft}
	\newboolean{showylabelright}
	\newboolean{showlegend}
	\begin{subfigure}[t]{0.325\columnwidth}
		\centering
		\setboolean{showylabelleft}{true}
		\setboolean{showylabelright}{true}
		\setboolean{showlegend}{true}
		\pgfplotstableread[col sep=comma]{graphics/matlab/BerlinTax/BerlinTaxCompetitive.csv}{\dataTable}
		\input{graphics/tikz/eScanTwo.tex}
		\caption{Identical \gls{abk:amod}s.}
		\label{fig:e7 competitive identical}
	\end{subfigure}
	\hfill  
	\begin{subfigure}[t]{0.325\columnwidth}
		\centering
		\setboolean{showylabelleft}{true}
		\setboolean{showylabelright}{true}
		\setboolean{showlegend}{false}
		\pgfplotstableread[col sep=comma]{graphics/matlab/BerlinTax/BerlinTaxCompetitive_Nonidentical.csv}{\dataTable}
		\input{graphics/tikz/eScanTwo.tex}
		\caption{Non-ident. \gls{abk:amod}s.}
		\label{fig:e7 competitive nonidentical}
	\end{subfigure}
	\hfill  
	\begin{subfigure}[t]{0.325\columnwidth}
		\centering
		\setboolean{showylabelleft}{true}
		\setboolean{showylabelright}{true}
		\setboolean{showlegend}{false}
		\pgfplotstableread[col sep=comma]{graphics/matlab/BerlinTax/BerlinTax.csv}{\dataTable}
		\input{graphics/tikz/eScanSingle.tex}
		\caption{Single \gls{abk:amod}.}
		\label{fig:e7 single}
	\end{subfigure}
	\caption{Impact of a service tax on the modal share and the \gls{abk:amod} operators' profit.}
	\label{fig:e7}
\end{figure}
Interestingly, we observe overproportionally decreasing profits for the larger operator in the two non-identical \gls{abk:amod} operator case. Here, the increasing tax share shrinks the set of potentially profitable customer trips equally for both operators. Accordingly, with increasing tax shares, the dominant influence on the equilibrium shifts from the fleet size to the competition between \gls{abk:amod} operators in a saturated market, which explains the steeper decrease in profit for the larger operator.

%% file: graphics/tikz/e1pieShareTwo.tex
\tikzstyle{legendPie} = [rectangle]

\begin{tikzpicture}[scale=0.5,font=\tiny]
    \pie[color={amod,amod1,publicTransit,walking},
         radius=\radiusPlot,
         rotate=30]{\amodfirstShare/, 
         			\amodsecondShare/,
                    \ptShare/, 
                    \walkingShare/}  
    
    \begin{scope}[shift={(-1.4,-\radiusPlot-0.5)}]
        \node[legendPie,fill=amod,draw=black,
              label={[font=\scriptsize]right:\acrshort{abk:amod} 1}] at (0.1,-0.5) {};
        \node[legendPie,fill=amod1,draw=black,
              label={[font=\scriptsize]right:\acrshort{abk:amod} 2}] at (0.1,-1.0) {};
        \node[legendPie,fill=publicTransit,draw=black,
              label={[font=\scriptsize]right:Public transport}] at (0.1,-1.5) {};
        \node[legendPie,fill=walking,draw=black,
              label={[font=\scriptsize]right:Walking}] at (0.1,-2.0) {};
    
    \begin{axis}[at={(4.1cm,3.5cm)},
                 anchor={north west},
                 width=2.8cm,height=6.0cm,
                 ymin=0,ymax=105,
                 ytick={0,10,20,30,40,50,60,70,80,90,100},
                 axis y line=left,
                 ylabel={\normalsize Profit share per trip [\si{\percent}]},
                 ylabel shift=-7pt,
                 xtick={1,2},
                 xticklabels={\acrshort{abk:amod} 1,\acrshort{abk:amod} 2},
                 xticklabel style={align=center,rotate=60,font=\normalsize},
                 x axis line style={opacity=0},
                 xtick style={draw=none},
                 ]
        \addplot [boxplot prepared={
                  draw direction=y,
                  lower whisker=\lowerQuantileFirst,
                  lower quartile=\midlowerQuantileFirst,
                  median=\midQuantileFirst,
                  upper quartile=\midhigherQuantileFirst,
                  upper whisker=\higherQuantileFirst},
                  color=black] coordinates {};
        \addplot [boxplot prepared={
                  draw direction=y,
                  lower whisker=\lowerQuantileSecond,
                  lower quartile=\midlowerQuantileSecond,
                  median=\midQuantileSecond,
                  upper quartile=\midhigherQuantileSecond,
                  upper whisker=\higherQuantileSecond},
                  color=black] coordinates {};
    \end{axis}
    \end{scope}

\end{tikzpicture}

%% file: graphics/tikz/e1pieShareSingle.tex
\tikzstyle{legendPie} = [rectangle]

\begin{tikzpicture}[scale=0.5,font=\tiny]
    \pie[color={amod,publicTransit,walking},
         radius=\radiusPlot,
         rotate=30]{\amodShare/, 
                    \ptShare/, 
                    \walkingShare/}  
    
    \begin{scope}[shift={(-1.4,-\radiusPlot-0.5)}]
        \node[legendPie,fill=amod,draw=black,
              label={[font=\scriptsize]right:\gls{abk:amod}}] at (0.5,-1.0) {};
        \node[legendPie,fill=publicTransit,draw=black,
              label={[font=\scriptsize]right:Public transport}] at (0.5,-1.5) {};
        \node[legendPie,fill=walking,draw=black,
              label={[font=\scriptsize]right:Walking}] at (0.5,-2.0) {};
    
    \begin{axis}[at={(4.1cm,3.5cm)},
                 anchor={north west},
                 width=2.8cm,height=6.0cm,
                 ymin=0,ymax=105,
                 ytick={0,10,20,30,40,50,60,70,80,90,100},
                 axis y line=left,
                 ylabel={\normalsize Profit share per trip [\si{\percent}]},
                 ylabel shift=-7pt,
                 xtick={1},
                 xticklabels={\normalsize \acrshort{abk:amod} 1},
                 xticklabel style={align=center,rotate=60},
                 x axis line style={opacity=0},
                 xtick style={draw=none},
                 ]
        \addplot [boxplot prepared={
                  draw direction=y,
                  lower whisker=13.51,
                  lower quartile=70.42,
                  median=77.93,
                  upper quartile=84.09,
                  upper whisker=99.50},
                  color=black] coordinates {};
    \end{axis}
    \end{scope}
\end{tikzpicture}

%% file: graphics/tikz/eScanTwo.tex

\ifthenelse{\boolean{showylabelleft}}
{\def\ylabelleft{Modal share [\%]}}
{\def\ylabelleft{}}

\ifthenelse{\boolean{showylabelright}}
{\def\ylabelright{Profit [\si{\usd\per\second}]}}
{\def\ylabelright{}}

\def\widthplot{4.0cm}

\begin{tikzpicture}[scale=0.8]
    \begin{axis}[width=\widthplot,height=4.0cm,
                 ymin=0,ymax=100,
                 xtick={\xticksPlot},
                 ytick={0,10,...,100},
                 x tick label style={/pgf/number format/set decimal separator={.}},
                 enlarge x limits=+0.05,
                 xlabel={\xlabelPlot},
                 ylabel={\ylabelleft},
                 ylabel shift=-9pt,
                 grid=major,
                 ybar stacked,
                 bar width=\barwidthPlot]
        \addplot[ybar stacked,ybar legend,mark=none,color=amod,fill=amod] table[x index=0,y index=1] {\dataTable};
        \addplot[ybar stacked,ybar legend,mark=none,color=amod,fill=amod1] table[x index=0,y index=2] {\dataTable};
        \addplot[ybar stacked,ybar legend,mark=none,color=publicTransit,fill=publicTransit] table[x index=0,y index=3] {\dataTable};
        \addplot[ybar stacked,ybar legend,mark=none,color=walking,fill=walking] table[x index=0,y index=4] {\dataTable};
    \end{axis}
    
    \begin{axis}[width=\widthplot,height=4.0cm,
                 ymin=0,
                 ymax=45,
                 ytick={\empty},
                 yticklabels={},
                 axis y line*=right,
                 axis x line=none,
                 enlarge x limits=+0.05,
                 ylabel={\ylabelright},
                 ylabel shift=-1pt,
                 legend style={at={(0.6,1.04)},
                               anchor=south west,
                               legend columns=3}]
        \addlegendimage{ybar stacked,ybar legend,mark=none,color=amod,fill=amod,xshift=0.5em}
        \addlegendentry{\acrshort{abk:amod} 1}
        \addlegendimage{ybar stacked,ybar legend,mark=none,color=amod1,fill=amod1,xshift=0.5em}
        \addlegendentry{\acrshort{abk:amod} 2}
        \addlegendimage{ybar stacked,ybar legend,mark=none,color=publicTransit,fill=publicTransit,xshift=0.5em}
        \addlegendentry{Public transport}
        \addlegendimage{ybar stacked,ybar legend,mark=none,color=walking,fill=walking,xshift=0.5em}
        \addlegendentry{Walking}
        \addplot[thick,black,mark=o,mark size=2.5pt] table[x index=0,y index=5] {\dataTable};
        \addlegendentry{Profit \acrshort{abk:amod} 1};
        \addplot[thick,black,mark=*,mark size=0.8pt,mark options={black,fill=black}] table[x index=0,y index=6] {\dataTable};
        \addlegendentry{Profit \acrshort{abk:amod} 2};
        \ifthenelse{\boolean{showlegend}}{}{\legend{};}
    \end{axis}
\end{tikzpicture}

%% file: graphics/tikz/eScanSingle.tex

\ifthenelse{\boolean{showylabelleft}}
{\def\ylabelleft{Modal share [\%]}}
{\def\ylabelleft{}}

\ifthenelse{\boolean{showylabelright}}
{\def\ylabelright{Profit [\si{\usd\per\second}]}}
{\def\ylabelright{}}

\def\widthplot{4.0cm}

\begin{tikzpicture}[scale=0.8]
    \begin{axis}[width=\widthplot,height=4.0cm,
                 ymin=0,ymax=100,
                 xtick={\xticksPlot},
                 ytick={0,10,...,100},
                 x tick label style={/pgf/number format/set decimal separator={.}},
                 enlarge x limits=+0.05,
                 xlabel={\xlabelPlot},
                 ylabel={\ylabelleft},
                 ylabel shift=-9pt,
                 grid=major,
                 ybar stacked,
                 bar width=\barwidthPlot]
        \addplot[ybar stacked,ybar legend,mark=none,color=amod,fill=amod] table[x index=0,y index=1] {\dataTable};
        \addplot[ybar stacked,ybar legend,mark=none,color=publicTransit,fill=publicTransit] table[x index=0,y index=2] {\dataTable};
        \addplot[ybar stacked,ybar legend,mark=none,color=walking,fill=walking] table[x index=0,y index=3] {\dataTable};
    \end{axis}
    
    \begin{axis}[width=\widthplot,height=4.0cm,
                 ymin=0,
                 ymax=45,
                 ytick={\empty},
                 yticklabels={},
                 axis y line*=right,
                 axis x line=none,
                 enlarge x limits=+0.05,
                 ylabel={\ylabelright},
                 ylabel shift=-1pt,
                 legend style={at={(-0.1,1.04)},
                               anchor=north west,
                               legend columns=3}]
        \addlegendimage{ybar stacked,ybar legend,mark=none,color=amod,fill=amod,xshift=0.5em}
        \addlegendentry{\acrshort{abk:amod}}
        \addlegendimage{ybar stacked,ybar legend,mark=none,color=publicTransit,fill=publicTransit,xshift=0.5em}
        \addlegendentry{Public transport}
        \addlegendimage{ybar stacked,ybar legend,mark=none,color=walking,fill=walking,,xshift=0.5em}
        \addlegendentry{Walking}
        \addplot[thick,black,mark=o,mark options={black,fill=black}] table[x index=0,y index=4] {\dataTable};
        \addlegendentry{Profit \acrshort{abk:amod}};
        \ifthenelse{\boolean{showlegend}}{}{\legend{};}
    \end{axis}
\end{tikzpicture}

%% file: chapters/7_conclusion.tex
\section{Conclusions}\label{sec:conclusion}
With this work, we focused on the interplay between \gls{abk:amod} fleets and public transportation. To this end, we developed a general methodological framework to model interactions among \glspl{abk:msp} and between \glspl{abk:msp} and customers.
Our framework combines a graph-theoretic network flow model with a game-theoretic approach to capture both the interactions between \glspl{abk:msp} and customers on the transportation market place and the constraints that result from the transportation network. We specified this framework for our application case, focusing on the interactions among two \gls{abk:amod} fleet operators, a municipality, and customers. We developed a computationally tractable second-order conic program to compute best responses, which can be used to find an equilibrium of the resulting game.
We applied our methodology to a real-world case study for the city of Berlin, and we presented results for various settings to identify major impact factors.

%% file: chapters/appendixnotation.tex
\section{Graph Theory}
\label{app:graphtheory}

\begin{definition}[Flow]
	\label{definition:flow}
	A flow $\flow{}$ is a pair $\flow{}=(\path{},\flowRate{})\in\pathsSet{\arcsSet{}}\times\positiveReals$, denoting the customer rate $\flowRate{}$ that uses a path $\path{}$. We introduce $\flowsSetGraph{\graph{}}$ as the set of all flows on $\graph{}$ and the projection operators $\projectionPathFlow{\flow{}}=\path{}$ and $\projectionRateFlow{\flow{}}=\flowRate{}$ mapping a flow to its path $\path{}$ and its flow rate $\flowRate{}$, respectively.
\end{definition}

\begin{definition}[Multigraph]
	\label{definition:multigraph}
	A multigraph $\graph{}$ is a quadruple $(\verticesSet{},\arcsSet{},\arcSource{},\arcTarget{})$ such that $\verticesSet{}$ is the set of vertices, $\arcsSet{}$ is the set of arcs, $\arcSource{}:\arcsSet{}\to\verticesSet{}$ assigns to each arc its source vertex, and $\arcTarget{}:\arcsSet{}\to\verticesSet{}$ assigns to each arc its sink vertex.
\end{definition}
	
\begin{definition}[Path]
	\label{definition:path}	
	We refer to a path of length $L\in\naturals$ as a set of distinct arcs $\{\arc{1},\ldots,\arc{L-1}\}$ for which there exists a set of exactly $L+1$ distinct vertices $\{\vertex{1},\ldots,\vertex{L}\}$ such that $\arcSource{}(\arc{i})=\vertex{i}$ and $\arcTarget{}(\arc{i})=\vertex{i+1}$ for all $i\in\{1,\ldots,L-1\}$. Note that by definition such a path cannot contain cycles. 
	Let $\pathsSet{\arcsSet{}}$ be the set of all paths.
\end{definition}

\begin{definition}[Origin and destination]
	\label{definition:origin and destination}
	Given a path $\path{}=\{\arc{1},\ldots,\arc{L-1}\}$ on $\graphDefinition{}$ we define the path origin and destination functions as $\originPath:\pathsSet{\arcsSet{}}\to\verticesSet{}$, $\destinationPath:\pathsSet{\arcsSet{}}\to\verticesSet{}$.
\end{definition}

\begin{definition}[Fully-connected]
	\label{definition:fully-connected}
	A multigraph $\graphDefinition{}$ is fully-connected if for all \mbox{$\vertex{1},\vertex{2}\in\verticesSet{}$}, $\vertex{1}\neq\vertex{2}$, there is $\arc{}\in\arcsSet{}$ such that $\arcSource{}(\arc{})=\vertex{1}$, $\arcTarget{}(\arc{})=\vertex{2}$.
\end{definition}

\begin{definition}[Fully-conn. graph]
	\label{definition:fully-connected graph}
	Let $\graphDefinition{}$ with $\verticesSet{}=\{\vertex{1},\ldots,\vertex{\verticesNumber{}}\}$. The fully-connected version of $\graph{}$ is $\graphDefinition[bar]{}$, where $\verticesSet[bar]{}=\verticesSet{}$, $\arcsSet[bar]{}=\{\arc{1,2},\arc{1,3},\ldots,\arc{\verticesNumber{},\verticesNumber{}-1}\}$, 	$\arcSource[bar]{}(\arc{i,j})= \vertex{i}$, and $\arcTarget[bar]{}(\arc{i,j}) = \vertex{j}$.
\end{definition}

\begin{definition}[Shortest path]
	Consider a multigraph $\graphDefinition{}$ and a non-negative function $f:\arcsSet{}\to\nonnegativeReals$. A path $\path{}$ is a shortest path between $\vertex{1}\in\verticesSet{}$ and $\vertex{2}\in\verticesSet{}$ if (\textit{i}) $\originPath(\path{})=\vertex{1}$ and $\destinationPath(\path{})=\vertex{2}$ and (\textit{ii}) it minimizes $\sum_{\arc{}\in\path{}}f(\arc{})$. We denote the set of shortest paths by $\shortestPath{}(\vertex{1},\vertex{2})$.
\end{definition}

%% file: chapters/appendixproofs.tex
\section{Proofs}
\label{app:proofs}

\def\proof{\noindent\hspace{2em}{\itshape Proof of \cref{lemma:reaction curve}: }}
\begin{proof}
The reaction curves arise directly from the definition. Indeed, let
\begin{equation*}
\begin{aligned}
	m_{i,1}
	&\coloneqq
	\begin{cases}
	-\frac{\demandRate{i}}{(\valueTimeMax-\valueTimeMin) \Delta t_i} & \text{if }\timePublicTransit{i}>\timeRoad{i}, \\
	-\frac{\demandRate{i}}{(\valueTimeMax-\valueTimeMin) \Delta t_i} & \text{if }\timePublicTransit{i}<\timeRoad{i},
	\end{cases}
	\\
	q_{i,1}
	&\coloneqq
	\begin{cases}
	-m_{i,1}\left(\valueTimeMax\Delta t_i + \pricePublicTransit{i}\right) & \text{if }\timePublicTransit{i}>\timeRoad{i},\\
	-m_{i,1}\left(-\valueTimeMin\Delta t_i + \pricePublicTransit{i}\right) & \text{if }\timePublicTransit{i}<\timeRoad{i},\\
	\end{cases}
\end{aligned}
\end{equation*}
with $\Delta t_i = |\timePublicTransit{i}-\timeRoad{i}|$, and let $m_{i,2}\coloneqq-\frac{1}{\bar\varepsilon}$, $q_{i,2}(\priceStrategy{-j}(\originVertex{i},\destinationVertex{i}))\coloneqq
	\frac{\demandRate{i}}{2}-m_{i,2}\priceStrategy{-j}(\originVertex{i},\destinationVertex{i})$, $m_{i,3}\coloneqq m_{i,2}+\frac{m_{i,1}}{2}$, and $q_{i,3}(\priceStrategy{-j}(\originVertex{i},\destinationVertex{i}))
	\coloneqq
	\frac{m_{i,1}\priceStrategy{-j}(\originVertex{i},\destinationVertex{i})+q_{i,1}}{2}-m_{i,3}\priceStrategy{-j}(\originVertex{i},\destinationVertex{i})$,  
where $\valueTimeMin$ and $\valueTimeMax$ are minimum and maximum values of the distribution of the value of time and $\bar\varepsilon$ is the width of the distribution of $\varepsilon_2$ (i.e., $\varepsilon_2$ is uniformly distributed between $-\bar\varepsilon/2$ and $\bar\varepsilon/2$). Recall that we assumed $\varepsilon_1=0$.
Using the definitions of argmax and of minimum and the computation of the expected value, we directly get~\eqref{eq:reaction curve}.
Clearly, $h_{i,j}$ is continuous and strictly decreasing. Hence, $p_{i,j}^\mathrm{min}$ and $p_{i,j}^\mathrm{max}$ are well-defined. Since $m_{i,1}$ and $q_{i,1}$ are both independent of $\priceStrategy{-j}$, $p_{i,j}^\mathrm{min}$ and $p_{i,j}^\mathrm{max}$ are upper bounded by $\max_{i\in\demandIndexSet}p_i^\ast<+\infty$, where $p_i^\ast$ is such that $m_{i,1}p_i^\ast+q_i=0$. 
Moreover, we directly see that $q_{i,2}$ and $q_{i,3}$ change linearly with $\priceStrategy{-j}$, which directly implies that $h_{i,j}$ is uniformly continuous with respect to $\priceStrategy{-j}$. 
Finally, since $h_{i,j}$ is continuous and $[p_{i,j}^\mathrm{min},p_{ij}^\mathrm{max}]$ is compact, $h_{i,j}^{-1}$ is continuous, too. So, $h_{i,j}$ is a homeomorphism. 
\end{proof}

\def\proof{\noindent\hspace{2em}{\itshape Proof of \cref{proposition:best response}: }}
\begin{proof}
We prove the proposition in four steps.
First, we observe that prices above $p_{i,j}^\mathrm{max}$ do not lead to larger profits. Indeed, such a price will result in $\reactionCurve{i}(\pathRoad{i,j})=0$, and therefore in the same revenue as the price $p_{i,j}^\mathrm{max}$.
Similarly, prices below $p_{i,j}^\mathrm{min}$ are never optimal. Indeed, they lead to the same reaction curve as $p_{i,j}^\mathrm{min}$, namely $\reactionCurve{i}(\pathRoad{i,j})=\demandRate{i}$, but result in a smaller profit due to a smaller revenue but equal costs. Hence, we can without loss of generality assume $\priceStrategy{j}(\originVertex{i},\destinationVertex{i})\in[p_{i,j}^\mathrm{min}, p_{i,j}^\mathrm{max}]$. 
Second, by the first step and \cref{lemma:reaction curve}, operator $j$ can optimize her profit by taking $x_i\coloneqq\reactionCurve{i}(\pathRoad{i,j})$ instead of the pricing strategies $\priceStrategy{j}$ as optimization variables. Indeed, there exists a bijection between $\reactionCurve{i}(\pathRoad{i,j})$ and $\priceStrategy{j}(\originVertex{i},\destinationVertex{i})$. Accordingly, the revenue of serving demand $\demand{i}$ becomes $x_i h_{i,j}^{-1}(x_{i})$. Using an epigraph formulation, we can equivalently express it through the variable $r_i$. Imposing $r_i\geq 0$ ensures that only non-negative prices are chosen.  
Third, we notice that the optimal profit is attained when the set of flows that serves the $i$\textsuperscript{th} demand coincides with the shortest path, computed by weighing each arc $\arc{}\in\arcsSet{j}$ with $\operationCostFunction{j}(\arc{})$. Hence, $\flowsSet{i}=\{f_{i,j}^\ast\}$. Indeed, any other set of flows would incur larger costs, and thus smaller profit. 
Fourth, combining all observations leads to the second-order conic program \eqref{eq:best response socp}. Then, best response results from ``inverting'' $x_i^\ast$. Formally, if $x_i^\ast>0$ then $\priceStrategy{j}(\originVertex{i},\destinationVertex{i})=h_{i,j}^{-1}(x_i^\ast)$. If $x_i^\ast=0$, then $\priceStrategy{j}(\originVertex{i},\destinationVertex{i})\in[p_{i,j}^\mathrm{max},+\infty)$. If no demand goes from $\originVertex{i}$ and $\destinationVertex{i}$, then $\priceStrategy{j}(\originVertex{i},\destinationVertex{i})\in\nonnegativeReals$.  
\end{proof}


For \cref{theorem:existence} we need the following well-known lemma.  

\begin{lemma}
\label{lemma:existence equilibria}
Identify $\priceStrategiesSet{j}$ with a subset of $\nonnegativeReals^{\cardinality{\verticesSet{j}}^2}$ and assume it is non-empty, compact, and convex. 
Define the best response map $\bestResponse{j}:\priceStrategiesSet{-j}\to 2^{\priceStrategiesSet{j}}$ of operator $j$ as $\bestResponse{j}(\priceStrategy{-j})\coloneqq\arg\max_{\priceStrategy{j}}\utility_j(\priceStrategy{j},\{\customersEquilibria{i}(\priceStrategy{j},\priceStrategy{-j})\}_{i=1}^{\demandNumber}).$
Assume that for all $j\in\operatorsIndexSet$ the best response map $\bestResponse{j}$  has closed graph and that $\bestResponse{j}(\priceStrategy{-j})$ is non-empty and convex for each $\priceStrategy{-j}\in\prod_{k=1,k\neq j}^{\operatorsNumber}\priceStrategiesSet{k}$.
Then, the game admits an equilibrium.
\end{lemma}


\def\proof{\noindent\hspace{2em}{\itshape Proof of \cref{lemma:existence equilibria}: }}
\begin{proof}
Define $\priceStrategiesSet{}\coloneqq\prod_{j=1}^{\operatorsNumber}\priceStrategiesSet{j}$ and $\bestResponse{}:\priceStrategiesSet{}\to 2^{\priceStrategiesSet{}}$ by 
    $\bestResponse{}(\priceStrategy{1},\ldots,\priceStrategy{\operatorsNumber})
    \coloneqq
    \left\{(\priceStrategy[bar]{1},\ldots,\priceStrategy[bar]{\operatorsNumber})
    \,|\,
    \priceStrategy[bar]{j}\in\bestResponse{j}(\priceStrategy{-j})\:\forall\,j=1,\ldots,\operatorsNumber\right\}.$
Since $\priceStrategiesSet{}$ is the Cartesian product of non-empty, convex, and compact sets, it is non-empty, convex, and compact.
Similarly, $\bestResponse{}$ has closed graph and $\bestResponse{}(\priceStrategy{1},\dots,\priceStrategy{\operatorsNumber})$ is non-empty and convex for all $\priceStrategy{1},\dots,\priceStrategy{\operatorsNumber}\in\priceStrategiesSet[]{}$.
Then, an equilibrium is a fixed point of $\bestResponse{}$. Its existence follow from Kakutani's fixed point theorem and is a standard result in game theory.
\end{proof}


\def\proof{\noindent\hspace{2em}{\itshape Proof of \cref{theorem:existence}: }}
\begin{proof}
It suffices to show that the assumptions of \cref{lemma:existence equilibria} hold true. Then, existence of an equilibrium follows directly.
First, we identify the space of pricing strategies with $\nonnegativeReals^{\cardinality{\verticesSet{j}}^2}$. Then, convexity and non-emptiness of $\priceStrategiesSet{j}$ follow straightforwardly.
We now show that there is no loss of generality in assuming compactness of $\priceStrategiesSet{j}$. Recall that by \cref{lemma:reaction curve} $p_{i,j}^\mathrm{max}$ is uniformly upper bounded by some $M>0$. Consider a modified version of the game whereby all prices are upper bounded by $M$. If this game possesses an equilibrium, then this equilibrium is an equilibrium in the original game as well, as showed in the first step of the proof of \cref{proposition:best response}.
Hence, $\priceStrategiesSet{j}$ can be assumed to be compact. 
Second, the proof of \cref{proposition:best response} shows that $\bestResponse{j}(\priceStrategy{-j})$ is non-empty and convex for each $\priceStrategy{-j}$. Indeed, $\bestResponse{j}(\priceStrategy{-j})$ consists of the cartesian product of sets of the form $\{p\}$ for appropriate $p\in\nonnegativeReals$, $[p_{ij}^\mathrm{max},M]$, or $[0,M]$. Note also that $\bestResponse{j}(\priceStrategy{-j})$ is closed, being the cartesian product of closed sets.  
Third, we need to establish that $\bestResponse{j}$ has closed graph. Since $\bestResponse{j}(\priceStrategy{-j})$ is closed, it suffices to prove that it is upper hemicontinuous; i.e., if $\priceStrategy{-j,n}\to\priceStrategy{-j}$ and $\priceStrategy{j,n}\to\priceStrategy{j}$ with $\priceStrategy{j,n}\in\bestResponse{j}(\priceStrategy{-j})$, then $\priceStrategy{j}\in\bestResponse{j}(\priceStrategy{-j})$.
The optimization problem is strongly convex in each $x_i$, so its optimal $x_i^\ast$ is unique and continuous with respect to the parameters $q_{1,i}$ and $q_{2,i}$. Moreover, \cref{lemma:reaction curve} implies that $q_{1,i}$ and $q_{2,i}$ are continuous in $\priceStrategy{-j}(\originVertex{i},\destinationVertex{i})$, and in particular that $p_{i,j}^\mathrm{max}$ and $p_{i,j}^\mathrm{min}$ are continuous in $\priceStrategy{-j}(\originVertex{i},\destinationVertex{i})$. Thus, $\bestResponse{j}(\priceStrategy{-j})$ results from the cartesian product of closed sets, whose limits are continuous in $\priceStrategy{-j}$. Hence, let $\priceStrategy{-j,n}\to\priceStrategy{-j}$ and $\priceStrategy{j,n}\to\priceStrategy{j}$ with $\priceStrategy{j,n}\in\bestResponse{j}(\priceStrategy{-j,n})$. Since all bounds of the involved sets are continuous in $\priceStrategy{-j}$, we only need to focus on the case when $\priceStrategy{j,n}\in\bestResponse{j}(\priceStrategy{-j,n})$ is of the form $[p,M]$ for all $n$ sufficiently large, but $\priceStrategy{j}\in\bestResponse{j}(\priceStrategy{-j})$ is a singleton. Yet, we can rule out this case by observing that a $\priceStrategy{-j,n}$ that results in this case does not exist, due to continuity and closedness of all sets. 
Then, \cref{lemma:existence equilibria} gives existence of an equilibrium.
\end{proof}

\def\proof{\noindent\hspace{2em}{\itshape Proof of \cref{theorem:equilibrium linear}: }}
\begin{proof}
We can see the single operator case as the special case, simply by letting $\priceStrategy{-j}(\originVertex{},\destinationVertex{})\to\infty$ for all $\originVertex{},\destinationVertex{}\in\verticesSet{-j}$. If so, $q_{2,i},q_{3,i}\to\infty$, which gives $r_{i}=x_i\cdot(x_i-q_{1,i})/m_{1,i}$. This simplifies the computation of the best response (and thus of the equilibium) to a quadratic program.
Since $m_{1,i}<0$, the objective function is strictly concave in $x_i$, which establishes the uniqueness of the optimal solution $x_i^\ast$. Thus, we can conclude that all equilibria result in the same reaction curves and profits.
\end{proof}